\documentclass[11pt]{article}

\usepackage[
  bookmarks=true,
  bookmarksnumbered=true,
  bookmarksopen=true,
  pdfborder={0 0 0},
  breaklinks=true,
  colorlinks=true,
  linkcolor=black,
  citecolor=black,
  filecolor=black,
  urlcolor=black,
]{hyperref}

\usepackage{graphicx}
\usepackage{color}
\usepackage[margin=1in]{geometry}
\usepackage{amsthm}
\usepackage{amsmath}
\usepackage{amssymb}
\usepackage{isomath}
\usepackage{enumitem}
\usepackage{algorithm}
\usepackage{pbox}
\usepackage{hhline}

\hypersetup{
  pdfauthor      = {Yannai A. Gonczarowski, Noam Nisan, Rafail Ostrovsky, and Will Rosenbaum},
  pdftitle       = {A Stable Marriage Requires Communication},
}

\title{A Stable Marriage Requires Communication}
\author{Yannai A. Gonczarowski\thanks{The Hebrew University of Jerusalem (Einstein Institute of Mathematics, Rachel \& Selim Benin School of Computer Science \& Engineering and Federmann Center for the Study of Rationality) and Microsoft Research. \mbox{\emph{E-mail}: \href{mailto:yannai@gonch.name}{yannai@gonch.name}}.
}
\and Noam Nisan\thanks{The Hebrew University of Jerusalem (Rachel \& Selim Benin School of Computer Science \& Engineering and Federmann Center for the Study of Rationality) and Microsoft Research. \mbox{\emph{E-mail}: \href{mailto:noam.nisan@gmail.com}{noam.nisan@gmail.com}}.
}
\and Rafail Ostrovsky\thanks{University of California Los Angeles (Department of Computer Science and Mathematics). \mbox{\emph{E-mail}: \href{mailto:rafail@cs.ucla.edu}{rafail@cs.ucla.edu}}.
}
\and Will Rosenbaum\thanks{Tel Aviv University (School of Electrical Engineering); work carried out while at University of California Los Angeles (Department of Mathematics). \mbox{\emph{E-mail}: \href{mailto:will.rosenbaum@gmail.com}{will.rosenbaum@gmail.com}}.
}}
\date{July 25, 2018}

\theoremstyle{plain}
\newtheorem{theorem}{Theorem}[section]
\newtheorem{lemma}{Lemma}[section]
\newtheorem{corollary}{Corollary}[section]
\newtheorem{proposition}{Proposition}[section]

\theoremstyle{definition}
\newtheorem{Definition}{Definition}[section]
\newtheorem{Example}{Example}[section]
\newtheorem{Remark}{Remark}[section]
\newtheorem{open}{Open Problem}[section]

\newcommand{\marriage}{\mu}
\newcommand{\marriageid}{\marriage_{\mathrm{id}}}
\newcommand{\allprefs}{\mathcal{P}}
\newcommand{\fullprefs}{\mathcal{F}}
\newcommand{\prefs}[1]{P_{#1}}
\newcommand{\married}[1]{#1_{\marriage}}
\newcommand{\fmopt}{f_{\mbox{\scriptsize $M$-Opt}}}
\newcommand{\marriagemarket}{\bigl(W,M,\prefs{W},\prefs{M}\bigr)}
\newcommand{\marriagemarketnumbered}{\bigl(W=\{w_1,\ldots,w_n\},\allowbreak M=\{m_1,\ldots,m_n\},\prefs{W},\prefs{M}\bigr)}

\newcommand{\distinctpairs}[1]{#1^{\underline{2}}}

\DeclareMathOperator{\DISJ}{DISJ}
\newcommand{\abs}[1]{\left|#1\right|}
\newcommand{\bfx}{\bar{x}}
\newcommand{\bfy}{\bar{y}}
\newcommand{\calM}{\mathcal{M}}
\newcommand{\dft}[1]{{\bfseries\emph{#1}}}
\newcommand{\e}{\varepsilon}
\newcommand{\Oh}{\mathcal{O}}
\newcommand{\set}[1]{\left\{#1\right\}}
\newcommand{\st}{\,\middle|\,}

\DeclareMathOperator{\bp}{\mathit{bp}}

\newcommand{\wrt}{with respect to}
\newcommand{\sut}{such that}
\newcommand{\wwlog}{without loss of generality}

\begin{document}
  \maketitle


\begin{abstract}
The Gale-Shapley algorithm for the Stable Marriage Problem is known to 
take $\Theta(n^2)$ steps to find a stable marriage in the worst case, 
but only $\Theta(n \log n)$ steps in the average case (with $n$ women 
and $n$ men). In 1976, Knuth asked whether 
the worst-case running time can be improved in a model of computation that 
does not require sequential access to the whole input. A partial negative 
answer was given by Ng and Hirschberg, who showed that 
$\Theta(n^2)$ queries are required in a model that allows certain natural 
random-access queries to the participants' preferences.
A significantly more general --- albeit slightly weaker --- lower bound
follows from Segal's general analysis of communication complexity, namely that $\Omega(n^2)$
Boolean queries are required in order to find a stable marriage, regardless of the set of allowed Boolean queries.

Using a reduction to the communication complexity of the disjointness problem,
we give a far simpler, yet significantly more powerful argument showing that $\Omega(n^2)$ Boolean queries of any type are indeed required for finding a stable --- or even an approximately stable --- marriage.
Notably, unlike Segal's lower bound, our lower bound generalizes also to (A)~randomized algorithms,
(B)~allowing arbitrary separate preprocessing of the women's preferences profile and 
of the men's preferences profile, (C)~several variants of the basic problem, such 
as whether a given pair is married in every/some stable marriage, and
(D)~determining whether a proposed marriage is stable or far from stable. In order to analyze ``approximately stable'' marriages,  we introduce the notion of ``distance to stability'' and provide an efficient algorithm for its computation.
\end{abstract}

\paragraph{Keywords:} stable marriage; stable matching; approximately stable; communication complexity; distance to stability.


  \section{Introduction}
  \label{sec:intro}

  In the classic Stable Marriage Problem \cite{GS62}, there are $n$~\dft{women}
and $n$ \dft{men}; each woman has a full preference order over 
the men and each man has a full preference order over the women.  The 
challenge is to find a \dft{stable marriage}: a one-to-one mapping between 
women and men that is stable in the sense that it contains no \dft{blocking 
pair}: a woman and man who mutually prefer each other over their current spouse 
in the marriage.  Gale and Shapley \cite{GS62} proved that such a stable 
marriage exists by providing an algorithm for finding one. Their algorithm 
takes $\Theta(n^2)$ steps\footnote{For a brief introduction to computer-science notation such as $\Theta(n^2)$, $\Omega(n^2$), and $\Oh(n^2)$, see Appendix~\ref{sec:asymptotic-notation}.} in the worst case \cite{GS62}, but only 
$\Theta(n\log n)$ steps in the average case, over independently and uniformly 
chosen preferences \cite{Wilson72}.

In 1976, Knuth~\cite{Knuth76} asked whether this quadratic worst-case running time 
can be improved upon. A related question was put forward in 1987 by Gusfield~\cite{Gusfield87}, who asked whether even \emph{verifying} the stability
of a proposed marriage can be done any faster. As the
input size here is quadratic in $n$, these questions only make sense in models 
that do not require sequentially reading the whole input, but rather provide 
some kind of random access to the preferences of the participants.

While Knuth's and Gusfield's questions arose from computational concerns, they also have tangible economic significance. In many real-world matching mechanisms, it is unreasonable to expect participants to provide their full preference list over all alternatives. This is not merely due to the effort of writing down an immense ordered list of alternatives, but also due to the sheer cognitive or physical effort of forming these preferences (for example, by conducting interviews). Formally, the process of forming or revealing one's preferences can be modeled as ``querying'' the individual's (perhaps implicitly defined) preferences. Each query consists of answering a single question about the individual's preferences that requires only a short response.\footnote{The requirement that responses are ``short'' rules out, for example, the possibility of a participant being asked their entire preferences with a single query. In our analysis, we consider Boolean queries, i.e.\ queries are ``yes/no'' questions. Other models (e.g.\ that of Ng and Hirschberg~\cite{HN90}) consider queries with slightly longer ($\Oh(\log n)$-bit) responses. This difference in the response size of allowed queries can only affect the complexity by at most a factor of the length (number of bits) of the longest response.} Thus, the inherent complexity of a marriage mechanism can be measured by the number of queries necessary to participate in the mechanism.

A partial answer to both Knuth's and Gusfield's questions was given by Ng and Hirschberg \cite{HN90}, 
who considered a model that allows two types of unit-cost
queries to the preferences of the participants: ``what is woman~$w$'s ranking of man~$m$?'' 
(and, dually, ``what is man~$m$'s ranking of woman~$w$?'') and ``which man 
does woman~$w$ rank at place~$k$?'' (and, dually, ``which woman does man~$m$ 
rank at place~$k$?''). In this model, they prove a tight $\Theta(n^2)$ lower 
bound on the number of queries that any deterministic algorithm that solves the stable 
marriage problem, or even verifies whether a given marriage is stable, must 
make in the worst case.
Chou and Lu~\cite{CL10} later showed that even if one is allowed to separately query each of the $\log n$ bits of the answer to queries such as ``which man 
does woman~$w$ rank at place~$k$?'' (and its dual query), $\Theta(n^2\log n)$ such Boolean queries are still required in order to deterministically find a stable marriage.

These results still leave two questions open. The first is 
whether some more powerful model may allow for faster algorithms. While many
``natural'' algorithms for stable marriage do fit into these models, there may 
be others that do not. Indeed, there exist problems for which ``computationally unnatural'' operations, such as various types of hashing, arithmetic operations, or even
``cognitively natural'' operations such as processing through a neural network, do give 
algorithmic speedups. Further, it may be the case that the participants' preferences are only defined implicitly by their actual input. For example, a participant's type could be a point in some (possibly high dimensional) geometric space, such that they prefer to be married to partners whose types are geometrically close to their own (see, for example, \cite{Bogomolnaia2007, Bhatnagar2008}). In this case, a natural query may be of the form ``what is your type's $k$th coordinate?''
Thus, it is of interest to ask whether an algorithm that queries the actual (geometric) input can be significantly more efficient than one that only queries the implicitly defined preferences.

The second question concerns randomized algorithms: can 
they do better than deterministic ones?  This question is especially 
fitting for the stable marriage problem as the \emph{expected} running time is known to be 
small when the preferences are chosen uniformly at random.\footnote{In particular, this would be the case if the expected running 
time could be made small for \emph{any} distribution on preferences, rather 
than just the uniform one.}  We give a negative answer to both hopes, as well 
as several other related problems, thereby showing that answering a wide variety of basic questions related to stable marriages requires a quadratic number of queries (that is, requires querying nearly the entire preference structure):

\begin{theorem}[Informal, see Corollaries~\ref{cor:asm} and~\ref{cor:related}]
\label{thm:query-informal}
Any randomized (or deterministic) algorithm that uses any type of Boolean 
queries to the women's and to the men's preferences to solve any of the following 
problems requires $\Omega(n^2)$ queries in the worst case:
\begin{enumerate}[label={(\alph*)},ref={\alph*}]
\item\label{thm:query-informal-find} finding an (approximately) stable marriage,\footnote{Our notion of ``approximately stable marriage'' is that the marriage shares many married couples with some stable marriage; see Definition~\ref{dfn:asm} for a formal definition and a discussion, and Section~\ref{sec:dist-to-stability} for proof of tractability of this notion.}
\item\label{thm:query-informal-verify} determining whether a given marriage is stable or far from stable,
\item\label{thm:query-informal-ms} determining whether a given pair of participants is contained in some/every stable marriage,
\item\label{thm:query-informal-sp} finding any $\e n$ pairs of participants that appear in some/every stable marriage.
\end{enumerate}
These lower bounds hold even if we allow arbitrary preprocessing of all the men's preferences 
and of all the women's preferences separately. The lower bound for Part~(\ref{thm:query-informal-find}) holds regardless 
of which (stable or approximately stable) marriage is produced by the algorithm.
\end{theorem}

Our proof of Theorem \ref{thm:query-informal} comes from a reduction to the well-known 
lower bounds for the disjointness problem \cite{KS92,Razborov92} in 
Yao's \cite{Yao79} model of two-party communication complexity 
(see \cite{KN97} for
a survey). We consider a scenario in which Alice holds the preferences of the 
$n$ women and Bob holds the preferences of the $n$ men, and show that each of the problems
from Theorem~\ref{thm:query-informal} requires the exchange of $\Omega(n^2)$ bits of communication
between Alice and Bob.

We note that Segal~\cite{Segal03} shows by a general argument that any deterministic or nondeterministic\footnote{We use the term ``nondeterministic,'' as is customary in computer science, to refer not to randomized (i.e.\ probabilistic) algorithms, but rather to algorithms that, roughly speaking, need only verify the correctness of the output (rather than search for it). The precise definition, which is out of scope for this paper, is not required in order to follow the main text of this paper. See~\cite{KN97} for a comparison of different models of communication complexity.} communication protocol among all $2n$ participants for finding a stable marriage requires $\Omega(n^2)$ bits of communication. Our argument for Theorem~\ref{thm:query-informal}(\ref{thm:query-informal-find}), in addition to being significantly simpler, generalizes Segal's result to account for randomized 
algorithms,\footnote{We remark that in general, there may be an exponential gap 
between deterministic, nondeterministic, and randomized communication complexity.}
and even when considering only two-party communication between Alice and Bob (essentially allowing arbitrary
communication within the set of women and within the set of men without cost). Furthermore, our lower bound holds even for merely determining whether a given marriage is 
stable or far from stable (Theorem~\ref{thm:query-informal}(\ref{thm:query-informal-verify})), as well as for the additional related problems described in Theorem~\ref{thm:query-informal}(\ref{thm:query-informal-ms},\ref{thm:query-informal-sp}).
These results immediately imply the same lower bounds 
for any type of Boolean queries in the original computation model, as Boolean queries can be simulated by a communication protocol.

As indicated above, Theorem~\ref{thm:query-informal}(\ref{thm:query-informal-find}), as well as the corresponding lower bound on the two-party communication complexity, holds not only for stable marriages but also for
approximately stable marriages.
In the context of communication complexity, Chou and Lu \cite{CL10} also study such a relaxation of the stable marriage problem in a restricted computational model in which
communication is non-interactive (a sketching model). Chou and Lu show that any (deterministic, non-interactive, $2n$-party) protocol that finds a marriage where only a constant fraction of participants
are involved in blocking pairs requires $\Theta(n^2 \log n)$ bits of communication. Our results are not directly comparable to these, as the two notions of approximate stability are not comparable. Furthermore, we use a significantly more general
computation model (randomized, interactive, two-party), but give a slightly weaker lower bound.

Our lower bound for verification complexity (given in Theorem \ref{thm:query-informal}(\ref{thm:query-informal-verify})) is 
tight. Indeed there exists 
a simple deterministic algorithm for verifying the stability of a proposed 
marriage, which requires $\Oh(n^2)$ queries even in the weak comparison 
model that allows only for queries of the form ``does woman~$w$ prefer man~$m_1$
over man~$m_2$?'' and, dually, ``does man~$m$ prefer woman~$w_1$ over 
woman~$w_2$?''\footnote{By simple batching, this verification algorithm can 
be converted into one that uses only $\Oh(\frac{n^2}{\log n})$ queries, each 
of which returns an answer of length $\log n$ bits (with each query still 
regarding the preferences of only a single participant). This highlights the 
fact that the lower bounds of \cite{HN90} crucially depend on
the exact type of queries allowed in their model.} We do not know whether the 
lower bound is tight also for finding a stable marriage (Theorem~\ref{thm:query-informal}(\ref{thm:query-informal-find})).
Gale and Shapley's algorithm uses $\Oh(n^2)$ queries in the worst case, 
but $\Oh(n^2)$~of these queries require answers of $\log n$ bits each. 
Thus, the algorithm requires a total of $\Oh(n^2 \log n)$ Boolean 
queries, or bits of communication. We do not know whether $\Oh(n^2)$ 
Boolean queries suffice for any algorithm. While the gap between Gale 
and Shapley's algorithm and our lower bound is small, we believe that it is
interesting, as the number of queries performed by the algorithm is exactly 
linear in the input encoding length. An even slightly sublinear algorithm 
would therefore be interesting.\footnote{Note that, as shown in 
Appendix~\ref{nondeterministic}, the \emph{nondeterministic} communication complexity
is $\Theta(n^2)$, so proving higher lower bounds for the deterministic or 
randomized case may be challenging.} We indeed do not have any $o(n^2 \log n)$ 
algorithm, even randomized and even in the strong two-party communication 
model, nor do we have any improved $\omega(n^2)$ lower bound, even for 
deterministic algorithms and even in the simple comparison model.\footnote{It is interesting to note that Bei \emph{et al.}~\cite{BeiChenZhang13} identify a similar gap for the stable marriage problem in a dramatically different computation model of trial and error.}

\begin{open}\label{findstab-comparisons}
Consider the Comparison model for stable marriage that only allows for queries 
of the form ``does man~$m$ prefer woman~$w_1$ over woman~$w_2$?'' and, dually,
``does woman~$w$ prefer man~$m_1$ over man~$m_2$?''.  How many such queries 
are required, in the worst case, to find a stable marriage?
\end{open}  


  \section{Model and Preliminaries}
  \label{sec:prelims}

  \subsection{The Stable Marriage Problem}

  \subsubsection{Full Preference Lists}\label{sec:full-prefs}

  For ease of presentation, we consider a simplified version of the model of 
Gale and Shapley~\cite{GS62}. Let $W$ and $M$ be disjoint finite sets, of 
\dft{women} and \dft{men}, respectively, such that $|W|=|M|$.

\begin{Definition}[Full Preferences]\leavevmode
  \begin{enumerate}
  \item A \dft{full preference list} over $M$ is a total ordering of $M$.
  \item A \dft{profile of full preference lists} for $W$ over $M$ is a 
    specification of a full preference list over~$M$ for each woman $w\in W$. 
    We denote the set of all profiles of full preference lists for~$W$ over 
    $M$ by $\fullprefs(W,M)$.
  \item Given a profile $\prefs{W}$ of full preference lists for $W$ over $M$, 
    a woman $w\in W$ is said to \dft{prefer a man $\mathbfit{m\in M}$ over a man 
      $\mathbfit{m'\in M}$}, denoted by $\mathbfit{m \succ_w m'}$, if $m$ precedes $m'$ on the 
    preference list of $w$. We say that $w \in W$ \dft{weakly prefers $\mathbfit{m}$ over $\mathbfit{m'}$} if either $m \succ_w m'$ or $m = m'$.
\end{enumerate}
We define full preference lists over $W$ and profiles of full preference 
lists for $M$ over $W$ analogously.
\end{Definition}

\begin{Definition}[Perfect Marriage]
  A \dft{perfect marriage} between $W$ and~$M$ is a one-to-one mapping 
  between $W$ and $M$.
\end{Definition}

\begin{Definition}[Marriage Market]
A \dft{marriage market} (with full preference lists) is a quadruplet $\marriagemarket$, where $W$ and $M$ are disjoint, $|W|=|M|$, $\prefs{W}\in\fullprefs(W,M)$ and $\prefs{M}\in\fullprefs(M,W)$.
\end{Definition}

\begin{Definition}[Stability]\label{def:stability}
  Let $\marriagemarket$ be a marriage market and let $\marriage$ be a perfect marriage (between $W$ and $M$).
\begin{enumerate}
\item
A pair $(w,m)\in W\times M$ is said to be a \dft{blocking pair} (in $\marriagemarket$) with respect to $\marriage$, if each of $w$ and $m$ prefer the other
over their spouse in $\marriage$.
\item
$\marriage$ is said to be \dft{stable} if no blocking pairs exist \wrt\ $\marriage$. Otherwise, $\marriage$ is said to be \dft{unstable}.
\end{enumerate}
\end{Definition}

  \subsubsection{Arbitrary Preference Lists}\label{sec:arbitrary-prefs}

  While our main results are phrased in terms of full preference lists and 
perfect marriages, some additional and intermediate results in Section
\ref{sec:smlb} and in the appendix deal with an extended model, 
which allows for preferences to specify ``blacklists'' (i.e.\ declare some 
potential spouses as unacceptable) and for marriages to specify that some 
participants remain single. (This model is nonetheless also a simplified 
version of that of \cite{GS62}.) A (not necessarily full) 
\dft{preference list} over $M$ is a totally-ordered subset of $M$. We once 
again interpret a preference list as a ranking, from best to worst, of 
acceptable spouses. We interpret participants absent from a preference list 
as declared unacceptable, even at the cost of remaining single. Analogously, 
a \dft{profile of preference lists} for~$W$ over $M$ is a specification of a
preference list over $M$ for each woman $w\in W$; we denote the set of all 
profiles of preference lists for $W$ over $M$ by 
$\allprefs(W,M)\supset\fullprefs(W,M)$. In this extended model, a woman
$w$ is said to \dft{prefer a man~$\mathbfit{m}$ over a man $\mathbfit{m'}$} not only when $m$ 
precedes~$m'$ on the preference list of $w$, but also when~$m$ is on the 
preference list of $w$ while $m'$ is not. Again, if we say that $w$ \dft{weakly prefers $\mathbfit{m}$ over $\mathbfit{m'}$} if either $w$ prefers~$m$ over $m'$ or $m = m'$. (We once again define preference 
lists and profiles of preference list for $M$ over~$W$ analogously.)

  A (not necessarily perfect) \dft{marriage} between $W$ and~$M$ is a 
one-to-one mapping between a subset of~$W$ and a subset of~$M$. Given a 
marriage $\marriage$, we denote the set of married women (i.e.\ the subset 
of $W$ over which $\marriage$ is defined) by $\married{W}$; we analogously 
denote the set of married men by~$\married{M}$. For a marriage $\marriage$
to be \dft{stable} (\wrt\ $\prefs{W}$ and $\prefs{M}$), we require not only that
no blocking pair exist with respect to it, but also that no participant $p\in W\cup M$ be
married to someone not on the preference list of $p$.

  We note that this model of arbitrary (not necessarily full) preference lists generalizes the model of full preference lists described in Section~\ref{sec:full-prefs}. Indeed, if the preference list of each participant happens to contain all participants of the opposite gender, then the two notions of stability agree. In particular, any marriage that is stable with 
respect to such preference lists prescribes that no participant remains single.

  \subsubsection{Known Results}

  We now survey a few known results regarding the stable marriage problem, 
which we utilize throughout this paper. For the duration of this section, let $\marriagemarket$ be a marriage market, defined either according to the definitions of Section~\ref{sec:full-prefs} or according those of Section~\ref{sec:arbitrary-prefs}.

\begin{theorem}[Gale and Shapley \cite{GS62}]
  \label{thm:gs62}
  A stable marriage between $W$ and $M$ always exists. Moreover, there exists 
an $M$-optimal stable marriage, i.e.\ a stable marriage
where each man weakly prefers his spouse in this stable marriage over his spouse in any other stable marriage.
\end{theorem}

Gale and Shapley \cite{GS62} provide an efficient algorithm for finding the $M$-optimal stable marriage. Their algorithm 
runs in $\Theta(n^2)$ steps in the worst case, performing a query of $\Theta(\log n)$ bits in each step. Hence, the Gale-Shapley algorithm queries $\Theta(n^2\log n)$ bits in the worst case.

\begin{theorem}[McVitie and Wilson \cite{MW71}]
  \label{thm:mw71}
  The $M$-optimal stable marriage is also the $W$-pessimal stable marriage, 
  i.e.\ every other stable marriage is weakly preferred over it by each woman.
\end{theorem}

\begin{corollary}[$W$-pessimal\,\&\,$M$-pessimal $\Rightarrow$ unique]
  \label{cor:pessimal-unique}
  If a stable marriage is both the $W$-pessimal stable marriage and the $M$-pessimal
stable marriage, then it is the unique stable marriage.
\end{corollary}

\begin{theorem}[Roth's Rural Hospitals Theorem \cite{Roth86}]
  \label{thm:rural-hospitals}
  $\married{W}$ (resp.\ $\married{M}$) is the same for every stable marriage
  $\marriage$.
\end{theorem}

  \subsubsection{Approximately-Stable Marriages}

  In this section, we describe a notion of an ``approximately stable marriage.'' For ease of presentation, we restrict ourselves to marriage markets with full preference lists (i.e.\ the
model described in Section~\ref{sec:full-prefs}).
We define an approximately stable marriage as a perfect marriage that shares many married pairs with some (exactly) stable (perfect) marriage. Our definition is a natural 
generalization of that of {\"U}nver \cite{Unver05} (who considers 
only marriage markets with unique stable marriages), but it appears
to be novel in its exact formulation. Our notion of approximate stability
 has the theoretical advantage of being derived from a metric on the set of
all perfect marriages between $W$ and $M$.

\begin{Definition}
\label{dfn:asm}
For any pair of perfect marriages $\marriage, \marriage'$ between $W$ and $M$ (where $|W|=|M|=n$), we define the \dft{divorce distance} between $\marriage$ and $\marriage'$ to be\footnote{Abusing notation, we identify a perfect marriage $\marriage$ with the set of married pairs $\set{\bigl(w_1, \marriage(w_1)\bigr), \bigl(w_2, \marriage(w_2)\bigr), \ldots}$. Thus, $\marriage \cap \marriage'$ is the set of pairs (couples) that are married in both $\marriage$ and $\marriage'$.}
\[
d(\marriage, \marriage') = n - \abs{\marriage \cap \marriage'}.
\]
Note that $d$ measures the minimum number of divorces required to convert $\marriage$ to $\marriage'$ (and vice versa). By abuse of notation, we denote the \dft{divorce distance to stability} of a perfect marriage $\marriage$ to be
\[
d(\marriage) = \min_{\marriage' \in \calM} d(\marriage, \marriage')
\]
where $\calM$ is the set of all stable perfect marriages between $W$ and $M$. Thus, $d(\marriage)$ is the minimum number of divorces required to convert $\marriage$ into a stable marriage.

  We say that a marriage $\marriage$ is \dft{$\mathbfit{(1 - \e)}$-stable} if $d(\marriage) \leq \e n$. We say $\marriage$ is \dft{$\mathbfit{\e}$-unstable} if $d(\marriage) > \e n$.
\end{Definition}

\begin{Example}
$d(\marriage)=0$ if and only if $\marriage$ is stable. Therefore, for $\e=0$ the concepts of $1$-stability and $0$-instability coincide precisely with (exact) stability and instability, respectively.  Letting $\e$ grow, $(1-\e)$-stability is a weaker requirement for larger values of $\e$, while $\e$-instability is a stricter requirement for larger values of $\e$.
\end{Example}

When only a single stable marriage exists (in this case, as noted above, our definition of approximate stability coincides with that of {\"U}nver \cite{Unver05}), efficiently computing the divorce distance to stability of a given perfect marriage $\marriage$ is straightforward: first use the Gale-Shapley algorithm to compute the $M$-optimal stable marriage (which in this case is the unique stable marriage), and then calculate the divorce distance between this (unique) stable marriage and $\marriage$.  Unfortunately, this computation fails to generalize. Brute-force computation of $d(\marriage)$ by iterating over all stable marriages is infeasible for general preferences, because the set of all stable marriages~$\calM$ can be exponentially large~\cite{Knuth76,IL86,GI89}. Nonetheless, by exploiting the combinatorial structure of $\calM$, we show that $d(\marriage)$ can still be efficiently computed (albeit in slightly slower time $\tilde{\Oh}(n^4)$). We describe an algorithm to this effect in Section~\ref{sec:dist-to-stability}. We believe this algorithm to be of independent interest.

The concept of divorce distance is perhaps most valuable in developing our understanding of the qualitative behavior of \emph{exactly} stable marriage mechanisms. Given that (as our results on exact stability show) finding an exactly stable marriage requires very high communication (or alternatively, a very large number of queries), one may consider dynamic mechanisms that refine their output over time until reaching a stable marriage. Such mechanisms would produce some initial (not necessarily stable) marriage after an initial stage of communication/queries, and then after additional communication/queries, adjust the marriage to form a stable marriage.\footnote{In fact, one may argue that real-world marriages are formed in a similar manner: since no person can meet (and form preferences regarding) all of their potential spouses, they meet a relatively small number of potential spouses and marry based on the information gathered until that point. However, they do not stop meeting new people at that point in time. Over time, they may divorce their spouses in favor of spouses that they find more suitable.} If the social cost of each divorce due to this adjustment is high, then we would want to minimize the number of divorces in the second stage of the mechanism. That is, we would seek a marriage with small divorce distance to stability in the first stage, and would seek to replace it with the closest stable marriage in the second stage. The analysis of Section~\ref{sec:dist-to-stability} implies that if such a first stage can be constructed, then a corresponding second stage can be implemented in a computationally efficient manner (using quadratically many queries, of course). Our main result regarding approximately stable marriages implies that such a first stage cannot be implemented using significantly less communication or fewer queries than finding an exactly stable marriage. We further discuss implications and interpretations of this result in Section~\ref{sec:commentary}.

\paragraph{Comparison with Blocking-Pairs Approximate Stability}

There does not appear to be consensus in the literature on precisely how to quantify the (in)stability of a marriage with respect to the participants' preferences. A common notion of approximate stability is the requirement for a marriage to have relatively few blocking pairs; see, e.g.\ \cite{EH08}. For example, one can define the \dft{blocking-pairs instability} of a marriage $\mu$, which we denote by $\bp(\mu)$, to be the number of blocking pairs, divided by $n\cdot(n-1)$ (where this divisor is the total number of ``potential blocking pairs'' in a perfect marriage, i.e.\ woman-man pairs who are not married to each other).
As we now observe, having small divorce distance to stability is a strictly more demanding condition than having small blocking-pairs instability --- that is, divorce distance to stability is a \emph{finer} measure of instability than blocking-pairs instability.

\begin{proposition}[Small Divorce Distance Implies, but is Not Implied by, Few Blocking Pairs]\label{prop:blocking-pairs}
Let $W=\{w_1,\ldots,w_n\}$ and $M=\{m_1,\ldots,m_n\}$.
\begin{enumerate}[label={(\alph*)},ref={\alph*}]
\item\label{prop:blocking-pairs-implies}
For every $\prefs{W}$ and $\prefs{M}$, if a marriage $\mu$ is $(1 - \e)$-stable for some $\e$, then $\bp(\mu) \leq 2 \e$
\item\label{prop:blocking-pairs-not-implied}
There exist full preference lists $\prefs{W}$ and $\prefs{M}$, and a perfect marriage $\mu$, such that $\mu$ induces only a single blocking pair (so $\bp(\mu) = \frac{1}{n(n-1)}$), yet $d(\mu) = n$ (so $\mu$ is not $(1-\e)$-unstable for any $\e>0$).
\end{enumerate}
\end{proposition}

\begin{proof}
For Part~(\ref{prop:blocking-pairs-implies}), suppose that $\mu$ is $(1-\e)$-stable for some $\e > 0$. Therefore, $d(\mu) = \e n$, and by the definition of $d(\mu)$ there is some stable marriage $\mu'$ such that $\mu$ and $\mu'$ share $(1 - \e) n$ pairs. Let $W'$ and $M'$ denote the sets of women and men, respectively, whose partners are the same in $\mu$ and $\mu'$, so that $\abs{W'} = \abs{M'} = (1 - \e)n$. Since $\mu'$ is stable, there are no blocking pairs $(w, m) \in W' \times M'$ with respect to $\mu'$, and therefore there are no such pairs with respect to $\mu$. Therefore, any blocking pair $(m, w)$ with respect to $\mu$ must have $w \notin W'$ or $m \notin M'$ (or both). Since there are only $\e n$ choices of such $w$ (resp.\ of such $m$), the total number of blocking pairs induced by $\mu$ is less than $2 \e n(n-1)$, and so $\bp(\mu) \leq 2 \e$, as required.

For Part~(\ref{prop:blocking-pairs-not-implied}), we define full preference lists as follows. Each woman $w_i$ prefers the men in order: $m_1 \succ_{w_i} m_2 \succ_{w_i} \cdots \succ_{w_i} m_n$. Similarly, each man $m_j$ for $j \geq 2$ prefers the women in order: $w_1 \succ_{m_j} w_2 \succ_{m_j} \cdots \succ_{m_j} w_n$. Finally, the preference list of $m_1$ is defined as follows:
  \[
  w_1 \succ_{m_1} w_n \succ_{m_1} w_2 \succ_{m_1} w_3 \succ_{m_1} \cdots \succ_{m_1} w_{n-1}.
  \]
  A straightforward argument shows that the marriage $\mu' = \set{(w_i, m_i) \st i \in [n]}$ is the unique stable marriage for these preference lists.\footnote{For example, one can prove by induction on $i$ that $(w_1, m_1), (w_2, m_2), \ldots, (w_i, m_i)$ must be in any stable marriage for $i = 1, 2, \ldots, n$.} Consider now the marriage $\mu$ defined by
  \[
  \mu = \set{(w_1, m_2), (w_2, m_3), \ldots, (w_{n-1}, m_n), (w_n, m_1)}.
  \]
  The pair $(w_1, m_1)$ is the only blocking pair induced by $\mu$ (thus $\bp(\mu) = \frac{1}{n(n-1)}$). However, $\mu \cap \mu' = \varnothing$, so $d(\mu) = n$. 
\end{proof}

The fact that the divorce distance is a finer measure of instability than blocking-pairs instability makes it easier for us to prove lower bounds on the complexity of finding $(1 - \e)$-stable marriages. (Indeed, we do not know if similar lower bounds apply to finding a marriage $\mu$ satisfying $\bp(\mu) \leq \e$; see the discussion in Section~\ref{sec:commentary}.)
Which of these two measures of instability is more appropriate, though, depends on the given context.
 One justification for the use of blocking-pairs instability is that $\bp(\mu)$ is exactly the probability that a uniformly random pair $(w, m) \in (W \times M) \setminus \mu$ is a blocking pair. Thus, if $\bp(\mu)$ is small, it is perhaps unlikely that any instability will be detected, say, by randomly probing potential blocking pairs.

 In many real-world matching markets, partnerships are formed with incomplete knowledge about the market, and the participants continue to learn about the market after an initial matching is formed. In such situations, it is reasonable to expect that participants will eventually discover instabilities if any are present in the initial matching. Depending on how the market responds to instabilities, the number of blocking pairs and the divorce distance can each be interpreted as a natural measure of the social cost of the instability. If the participants are obligated to remain in the initial matching, then blocking-pairs instability can be viewed as quantifying the total ``suffering'' of participants due to their awareness of being part of a blocking pair. If, on the other hand, participants are allowed to act on the instability, but must pay some cost to break existing partnerships, then the divorce distance to stability quantifies the (minimum) cost for achieving stability.

  \subsection{Communication Complexity}

  We work in Yao's~\cite{Yao79} model of two-party communication complexity 
(see \cite{KN97} for a survey). Consider a scenario where two agents, Alice 
and Bob, hold values $x$ and $y$, respectively, and wish to collaborate in 
performing some computation that depends on both $x$ and $y$. Such a computation typically requires the exchange of some data 
between Alice and Bob. The \dft{communication cost} of a given protocol 
(i.e.\ distributed algorithm) for such a computation is the number of bits 
that Alice and Bob exchange under this protocol in the worst case (i.e.\ for 
the worst $(x,y)$); the \dft{communication complexity} of the computation that Alice and Bob wish to perform
is the lowest communication cost of any protocol for this computation.
Generalizing, we also consider \dft{randomized} communication complexity, 
defined analogously using randomized protocols that for every given fixed 
input, produce a correct output with probability at least 
$\frac{2}{3}$.\footnote{The results of this paper hold verbatim even if the 
constant~$\frac{2}{3}$ is replaced with any other fixed probability $p$ with
$\frac 1 2 < p \le 1$.
}

  Of particular interest to us is the disjointness function, $\DISJ$. 
Let $n\in\mathbb{N}$ and let Alice and Bob hold subsets $A, B \subseteq [n]$, respectively. The value of the 
disjointness function is $1$ if $A \cap B = \varnothing$, and $0$ otherwise. 
We can also consider $\DISJ$ as a Boolean function by identifying $A$ and $B$ with their respective characteristic vectors $\bar{x}=(x_i)_{i=1}^n$ and $\bar{y}=(y_i)_{i=1}^n$, defined by ~~
\mbox{$x_i = 1 \iff i \in A$} ~~and~~ \mbox{$y_j = 1 \iff j \in B$}.
Thus, we can express $\DISJ$ using the Boolean formula
\mbox{$\DISJ(\bar{x}, \bar{y}) = \neg \bigvee_{i = 1}^n (x_i \wedge y_i)$}.
All of our results heavily rely on the following result of Kalyanasundaram and Schintger~\cite{KS92} (see also Razborov~\cite{Razborov92}):

\begin{theorem}[Communication Complexity of $\DISJ$~\cite{KS92,Razborov92}]
Let $n\in\mathbb{N}$.
The randomized (and deterministic) communication complexity of calculating $\DISJ(\bar{x},\bar{y})$, where $\bar{x}\in\{0,1\}^n$ is held by Alice and $\bar{y}\in\{0,1\}^n$ is held by Bob, is $\Theta(n)$. Further, this lower bound holds even for unique disjointness, i.e.\ if we require that the inputs $\bar{x}$ and $\bar{y}$ are either disjoint or \dft{uniquely intersecting}: $\abs{\bar{x} \cap \bar{y}} \leq 1$. 
  \label{thm:disj}
\end{theorem}

  Our results regarding lower bounds on communication complexities all follow from defining suitable 
\dft{embeddings} of $\DISJ$ into various problems regarding stable marriages, i.e.\ mapping $\bar{x}$ and~$\bar{y}$ into suitable marriage markets (more specifically, mapping $\bar{x}$ into $\prefs{W}$ and $\bar{y}$ into $\prefs{M}$), such that
finding a stable marriage (or solving any of the other problems from Theorem~\ref{thm:query-informal}) reveals the value of $\DISJ$.
Some of our proofs (namely those presented in Section~\ref{sec:general-proofs}) indeed assume that the input to $\DISJ$ satisfies $\abs{\bar{x}\cap\bar{y}} \leq 1$.

  \section{Main Results}
  \label{sec:summary}

The survey of our main results begins with the result that directly implies our answer to Knuth's question~\cite{Knuth76} regarding the necessity of quadratic complexity for finding a stable marriage.

\begin{theorem}[Communication Complexity of Finding an Approximately-Stable Marriage]\label{thm:asm}
Let $\marriagemarketnumbered$ be a marriage market with full preference lists, where $\prefs{W}$ is held by Alice and $\prefs{M}$ is held by Bob, and let $0 \leq \e < \frac{1}{2}$.
The randomized (and deterministic) communication complexity of finding a $(1-\e)$-stable marriage in $\marriagemarket$ is $\Omega(n^2)$.
\end{theorem}

We note that in particular, setting $\e=0$ in Theorem~\ref{thm:asm} shows that finding an (exactly) stable marriage requires $\Omega(n^2)$ communication. Recall, though, that Knuth's question is phrased not with respect to communication complexity, but with respect to computational complexity. To answer Knuth's question, we reinterpret Theorem~\ref{thm:asm} using a standard argument that shows that the query complexity of a problem (that is, the worst-case number of Boolean queries any algorithm that solves this problem must make) is at least its communication complexity. 

\begin{corollary}[Query Complexity of Finding an Approximately-Stable Marriage]\label{cor:asm}
Let $\marriagemarketnumbered$ be a marriage market with full preference lists, and let $0 \leq \e < \frac{1}{2}$.
Any randomized (or deterministic) algorithm that uses any type of Boolean queries to the women's and (separately) to the men's preferences to find a $(1-\e)$-stable marriage requires $\Omega(n^2)$ queries in the worst case.
\end{corollary}

To see how Corollary~\ref{cor:asm} follows from Theorem~\ref{thm:asm}, suppose there is a randomized algorithm~$A$ that computes a $(1 - \e)$-stable marriage using $B$ Boolean queries to the women and men. We will use $A$ to construct a $B$-bit communication protocol for the approximate stable marriage problem. The protocol works as follows. Alice and Bob both simulate~$A$. Whenever $A$ queries the women's preferences, Alice sends the result of the query to Bob (since Alice knows the women's preferences). Symmetrically, whenever $A$ queries 
the men's preferences, Bob sends the result of the query to Alice. This protocol uses $B$ bits of communication. Thus, by Theorem \ref{thm:asm}, we must have $B = \Omega(n^2)$, as desired.

Having given an answer (up to a factor of $\log n$; see Open Question~\ref{findstab-comparisons} in the introduction) to Knuth's question on the complexity of finding a stable marriage, one may ask how robust this answer is to variations in the studied problem. We have already shown that this answer is robust in one sense: the same lower bound applies also to finding approximately stable marriages. This shows that the hardness is not a consequence of some delicate phenomenon driven by the exact stability of the desired outcome. (Since any stable marriage is also $(1-\e)$-stable, then allowing also approximately-stable marriage introduces additional possible solutions, so could have conceivably eased the task.) Another aspect of the robustness of our lower bound relates to Gusfield's question~\cite{Gusfield87} regarding verifying the stability of a marriage: is it any easier to verify the correctness of a single concretely-proposed solution to the stable marriage problem? Setting an even higher bar of robustness, one can ask the following variant of Gusfield's question: say that we are given a marriage that is not a ``borderline stable'' marriage. That is, this marriage is guaranteed to be either exactly stable or very far from stable. Would it still be as hard to differentiate between these two cases assuming all ``borderline'' cases have been ruled out? Generalizing Knuth's question in a different direction, one may also wonder if answering highly localized queries about the stable marriages in the market is any easier than finding a stable marriage. For example, is it easier to determine if a given woman-man pair is married in some (or in every) stable marriage?\footnote{See Appendix~\ref{isaremarried} for analogous results also for another natural localized query that arises in the context of incomplete preference lists: determining whether a given participant is single in some (equivalently, in every) stable marriage.} We complement Theorem~\ref{thm:asm} by showing that our quadratic lower bound applies to any and all of the above, and even more.

\begin{theorem}[Communication Complexity of Stable-Marriage Related Problems]
\label{thm:related}
Let $\marriagemarketnumbered$ be a marriage market with full preference lists, where $\prefs{W}$ is held by Alice and $\prefs{M}$ is held by Bob.
The randomized (and deterministic) communication complexity of solving any of the following problems is $\Omega(n^2)$.
\begin{enumerate}[label={(\alph*)},ref={\alph*}]
\item\label{thm:related-verify} determining whether a given marriage $\marriage$ is stable or $\e$-unstable, for fixed $\e$ with $0\le\e<1$.
\item\label{thm:related-ms} determining whether a given pair of participants is contained in some/every stable marriage.
\item\label{thm:related-sp} finding any $\e n$ pairs that appear in some/every stable marriage, for fixed $\e$ with $0<\e\le 1$.
\end{enumerate}
\end{theorem}

Similarly to above, we note that setting $\e=0$ in Theorem~\ref{thm:related}(\ref{thm:related-verify}) shows that verifying that a given marriage is (exactly) stable requires $\Omega(n^2)$ communication. We also note that unlike the case of Theorem~\ref{thm:asm}, where there is still a gap of $\log n$ that prevents us from completely understanding the communication complexity of finding a stable marriage (see Open Question~\ref{findstab-comparisons} in the introduction), this $\Omega(n^2)$ lower-bound on the communication complexity of verifying the stability of a given marriage is tight: exhaustively querying all pairs of participants to na{\"i}vely check for the existence of a blocking  pair requires $\Theta(n^2)$ bits of communication in the worst case.\footnote{While Theorem~\ref{thm:related}(\ref{thm:related-verify}) is phrased so that the marriage $\marriage$ is known by both Alice and Bob before the protocol commences, this tight bound still holds even if only one of them knows $\marriage$, as the straightforward way of encoding a marriage between $W$ and $M$ requires $\Oh(n\log n)$ bits.}
For the other problems analyzed in Theorem~\ref{thm:related}, a similar gap of $\log n$ remains: it is possible to show that $\Oh(n^2\log n)$ communication suffices, due to a deterministic algorithm given by Gusfield~\cite{Gusfield87} for enumerating all pairs that belong to at least one stable marriage in $\Oh(n^2\log n)$ Boolean queries, but the question of a tight bound for these problems remains open as well.

As with Theorem~\ref{thm:asm} yielding Corollary~\ref{cor:asm}, the same standard argument gives a query-complexity analogue of Theorem~\ref{thm:related}:

\begin{corollary}[Query Complexity of Stable-Marriage Related Problems]
\label{cor:related}
Let $\marriagemarketnumbered$ be a marriage market with full preference lists.
Any randomized (or deterministic) algorithm that uses any type of Boolean 
queries to the women's and (separately) to the men's preferences to solve any of the following 
problems requires $\Omega(n^2)$ queries in the worst case:
\begin{enumerate}[label={(\alph*)},ref={\alph*}]
\item\label{cor:related-verify} determining whether a given marriage $\marriage$ is stable or $\e$-unstable, for fixed $\e$ with $0\le\e<1$.
\item\label{cor:related-ms} determining whether a given pair of participants is contained in some/every stable marriage.
\item\label{cor:related-sp} finding any $\e n$ pairs that appear in some/every stable marriage, for fixed $\e$ with $0<\e\le 1$.
\end{enumerate}
\end{corollary}

The variety of the problems covered in Theorems~\ref{thm:asm} and~\ref{thm:related}
conveys the main message of this paper: Roughly speaking, answering practically any meaningful basic question regarding a stable marriage in a marriage market, requires a quadratic number of queries, i.e.\ nearly amounts to querying the entire preference structure.

The proofs of Theorems~\ref{thm:asm} and~\ref{thm:related} are given in Section~\ref{sec:proofs}. As we show there, all of these results can be proven by carefully embedding disjointness into a specially crafted marriage market that is described in Section~\ref{sec:embed}. Before describing this embedding in full detail, we first describe in Section~\ref{sec:smlb} a far more straightforward embedding of disjointness into a simpler marriage market. This embedding is just powerful enough to prove that verifying the (exact) stability of a given marriage requires $\Omega(n^2)$ communication. We then add slightly more detail to this construction in order to prove that finding an (exactly) stable marriage also requires $\Omega(n^2)$ communication. This gradual introduction to our machinery allows us to first expose the core of the argument that drives all of our main results (the fundamental relation between disjointness and stability) with only minimal surrounding detail. The simpler argument further allows readers who are unfamiliar with embeddings and disjointness to familiarize themselves with more straightforward examples of such techniques before delving into the intricacies of the general embedding.\footnote{Furthermore, the proof in Section~\ref{sec:smlb} of the lower bound for finding a stable marriage includes a novel application of the Rural Hospitals Theorem (Theorem~\ref{thm:rural-hospitals}) to facilitate the embedding of a decision problem (disjointness with no assumption of unique intersection) into the studied search problem, which we believe may be of independent interest.}

\section{Warm Up: Lower Bounds for Exact Stability}
  \label{sec:smlb}

  In this section, we give direct proofs of our lower bounds for the communication complexity of finding or verifying an exactly stable marriage. We
prove these lower bounds by embedding suitably large instances of $\DISJ$ into the problems of finding a stable marriage or verifying the stability of some marriage. Thus, by Theorem \ref{thm:disj} we obtain the desired lower bounds on the corresponding communication complexities. We note that the construction given in this section does not assume the input to $\DISJ$ to be uniquely intersecting.

\begin{Definition}
Let $n\in\mathbb{N}$. We denote the set of pairs of distinct elements of $\{1,\ldots,n\}$ by $\distinctpairs{[n]}=\set{(i,j) \in \{1,\ldots,n\}^2 \mid i \ne j}$.
We note that $\bigl|\distinctpairs{[n]}\bigr|=n\cdot(n-1)$.
\end{Definition}

For the duration of this section, let $n\in\mathbb{N}$, and let $W=\{w_1,\ldots,w_n\}$ and $M=\{m_1,\ldots,m_n\}$ be 
disjoint sets such that $|W|=|M|=n$. Let $\marriageid$ be the perfect marriage 
in which $w_i$ is married to $m_i$ for every $i$.  To lower-bound the communication complexity of verifying the stability of a given marriage, we embed disjointness into verification of stability.

\begin{lemma}[Disjointness $\hookrightarrow$ Verifying Stability]
\label{disj-in-isstab}
There exist
functions
$\prefs{W}:\{0,1\}^{\distinctpairs{[n]}}\rightarrow \fullprefs(W,M)$ and $\prefs{M}:\{0,1\}^{\distinctpairs{[n]}}\rightarrow \fullprefs(M,W)$ \sut\ for every $\bar{x}=(x^i_j)_{(i,j)\in\distinctpairs{[n]}}\in\{0,1\}^{\distinctpairs{[n]}}$ and $\bar{y}=(y^i_j)_{(i,j)\in\distinctpairs{[n]}}\in\{0,1\}^{\distinctpairs{[n]}}$, the following are equivalent.
\begin{itemize}
\item
$\marriageid$ is stable \wrt\ $\prefs{W}(\bar{x})$ and $\prefs{M}(\bar{y})$
\item
$\DISJ(\bfx,\bfy)=1$.
\end{itemize}
\end{lemma}

\begin{proof}
To define $\prefs{W}(\bar{x})$, for every $i$
we define the preference list of $w_i$ to consist of all $m_j$ \sut\ $x^i_j=1$, in arbitrary order (say, sorted by $j$), followed by $m_i$, followed by all other men in arbitrary order.
Similarly, to define $\prefs{M}(\bar{y})$, for every $j$
we define the preference list of $m_j$ to consist of all $w_i$ \sut\ $y^i_j=1$, in arbitrary order (say, sorted by $i$), followed
by $w_j$, followed by all other women arbitrary order.

$\marriageid$ is unstable \wrt\ $\prefs{W}(\bar{x})$ and $\prefs{M}(\bar{y})$
$\Longleftrightarrow$ there exist $(i,j) \in \distinctpairs{[n]}$ \sut\ $m_j \succ_{w_i} m_i$ and $w_i \succ_{m_j} w_j$ $\Longleftrightarrow$
there exist $(i,j) \in \distinctpairs{[n]}$ \sut\ $x^i_j=1$ and $y^i_j=1$ $\Longleftrightarrow$ $\DISJ(\bfx,\bfy)=0$.
\end{proof}

\begin{Remark}\label{isstab-in-disj}
A similar argument may be used to embed verification of stability back in disjointness.
\end{Remark}

To lower-bound the communication complexity of finding a stable marriage, we embed disjointness into finding a stable marriage through the intermediate 
problem of finding a stable marriage \wrt\ arbitrary (i.e.\ not necessarily full) preference lists. This embedding of disjointness into finding a stable marriage requires some more care than the above embedding into verification of stability, since unlike in Lemma~\ref{disj-in-isstab} where we have embedded a decision problem (i.e.\ a problem for which the answer is Boolean) into a decision problem, here we embed a decision problem into a search problem (in this case, a problem for which the answer is a marriage). While in Section~\ref{sec:general-proofs} we will mitigate this obstacle by assuming that the instance of disjointness that we embed is uniquely intersecting, here we overcome this obstacle via a novel application of the Rural Hospitals Theorem (Theorem~\ref{thm:rural-hospitals}), which we believe may be of independent interest.

\begin{lemma}[Disjointness $\hookrightarrow$ Finding a Stable Marriage (Arbitrary Preferences)]
\label{disj-in-findstabbl}
There exist functions
$\prefs{W}:\{0,1\}^{\distinctpairs{[n]}}\rightarrow \allprefs(W,M)$ and 
$\prefs{M}:\{0,1\}^{\distinctpairs{[n]}}\rightarrow \allprefs(M,W)$ \sut\
for every $\bfx=(x^i_j)_{(i,j)\in\distinctpairs{[n]}}\in\{0,1\}^{\distinctpairs{[n]}}$ 
and $\bfy=(y^i_j)_{(i,j)\in\distinctpairs{[n]}}\in\{0,1\}^{\distinctpairs{[n]}}$, 
both of the following hold.
\begin{enumerate}[label={\alph*.}]
\item
If $\DISJ(\bfx, \bfy) = 1$, then $\marriageid$ 
is the unique stable marriage with respect to $\prefs{W}(\bfx)$ and 
$\prefs{M}(\bfy)$.
\item
If $\DISJ(\bfx, \bfy) = 0$, then $\marriageid$ 
is unstable with respect to $\prefs{W}(\bfx)$ and $\prefs{M}(\bfy)$.
\end{enumerate}
\end{lemma}

\begin{proof}
To define $\prefs{W}(\bfx)$, for every $i$ we define the preference list of 
$w_i$ to consist of all $m_j$ \sut\ $x^i_j=1$, in arbitrary order (say, sorted 
by $j$), followed by $m_i$ (with all other men absent). Similarly, to define 
$\prefs{M}(\bfy)$, for every $j$ we define the preference list of $m_j$ to 
consist of all $w_i$ \sut\ $y^i_j=1$, in arbitrary order (say, sorted by~$i$), 
followed by $w_j$ (with all other women absent).

We first show that $\marriageid$ is stable with respect to $\prefs{W}(\bfx)$ 
and $\prefs{M}(\bfy)$ iff $\DISJ(\bfx, \bfy) = 1$.
 Indeed, since every 
participant is married by $\marriageid$ to someone on their preference list, 
we have: 

$\marriageid$ is unstable with respect to\ $\prefs{W}(\bfx)$ and 
$\prefs{M}(\bfy)$ 
$\iff$ there exist $(i,j) \in \distinctpairs{[n]}$ such that 
$m_j \succ_{w_i} m_i$ and $w_i \succ_{m_j} w_j$ 
$\iff$ there exist $(i,j) \in \distinctpairs{[n]}$ such that $x^i_j=1$ and 
$y^i_j=1$ 
$\iff$ $\DISJ(\bfx, \bfy) = 0$.

It remains to show that if $\marriageid$ is stable with respect to 
$\prefs{W}(\bfx)$ and $\prefs{M}(\bfy)$, then it is the unique stable 
marriage with respect to these profiles of preference lists. For the remainder 
of the proof assume, therefore, that $\marriageid$ is stable (with respect to 
$\prefs{W}(\bfx)$ and $\prefs{M}(\bfy)$). Let $\marriage$ be a stable 
marriage (with respect to these profiles of preference lists). As 
$\marriageid$ is stable and perfect, by Theorem \ref{thm:rural-hospitals},
since $\marriage$ is stable, it is perfect as well. Therefore, each 
$p \in W\cup M$ is married by $\marriage$ to someone on the preference list of 
$p$, and so $p$ weakly prefers~$\marriage$ over~$\marriageid$, as in the 
latter $p$ is married to the last person on the preference list of $p$. Thus, 
$\marriageid$ is both the \mbox{$W$-pessimal} stable marriage and the $M$-pessimal stable 
one, and so, by Corollary \ref{cor:pessimal-unique}, $\marriageid$ is the unique 
stable marriage.
\end{proof}

Our lower bound on the communication complexity of finding a stable marriage with respect to full preference lists follows from Lemma~\ref{disj-in-findstabbl} by showing
that we can embed the problem of finding a stable marriage with respect to possibly-partial
preference lists into finding a stable marriage with respect to full preference lists. See
Appendix~\ref{sec:arbitrary-complete} for details. 

The techniques used to
prove Lemmas~\ref{disj-in-isstab} and~\ref{disj-in-findstabbl} can 
also be used to prove Theorem~\ref{thm:related}(\ref{thm:related-ms}) --- see Appendix~\ref{isaremarried}.
Although Theorem~\ref{thm:related}(\ref{thm:related-ms}) shows that determining the marital status of a
fixed pair~$(w, m)$ requires $\Omega(n^2)$ communication, we do not know how
to prove a similar lower bound for finding \emph{some} married couple (see Open Problem~\ref{open:single-pair} in Section~\ref{sec:commentary}). In 
the next section, we however show a weaker related result, namely that finding any constant fraction of the couples married in a stable marriage
requires $\Omega(n^2)$ communication (Theorem~\ref{thm:related}(\ref{thm:related-sp})). This result stems from a more elaborate embedding, which also yields our main results regarding approximate stability (Theorem~\ref{thm:asm} and Theorem~\ref{thm:related}(\ref{thm:related-verify})).

  \section{General Proof of Main Results}\label{sec:general-proofs}

  \subsection{Embedding Disjointness into Preferences}
  \label{sec:embed}

  Similarly to the proofs given in Section~\ref{sec:smlb}, the proofs of the remaining results from Section~\ref{sec:summary} follow from 
embedding suitably large instances of $\DISJ$ into various problems regarding (approximately) stable marriages.
In order to prove these remaining results, we reconstruct the embeddings to have the property that small changes in the 
participants' preferences yield very large changes in the global structure of the stable marriages for these preferences. Informally, we construct the preferences so that resolving blocking pairs resulting from such small changes in participants' preferences creates large rejection chains that ultimately affect most married couples.

  \subsubsection{Preference Description}
  \label{sec:pref-description}

  Let $n\in\mathbb{N}$ and let $W$ and $M$ be disjoint \sut\ $|W|=|M|=n$.
  We divide the participants into three sets: \dft{high}, \dft{mid} and \dft{low}, which we denote 
$W_h$, $W_m$ and $W_l$ respectively for the women and $M_h$,~$M_m$~and~$M_l$ respectively for the 
men. These sets have sizes
\begin{alignat*}{5}
  &\abs{W_h} &\:=\:& \abs{M_h} &\:=\:& \tfrac 1 2 \delta n\\
  &\abs{W_m} &\:=\:& \abs{M_m} &\:=\:& \tfrac 1 2 (1 - \delta) n\\
  &\abs{W_l} &\:=\:& \abs{M_l} &\:=\:& \tfrac 1 2 n
\end{alignat*}
  where $\delta$ is a parameter with $0 < \delta \le 1$, to be chosen later. The low and mid participants 
preferences will be fixed, while we will use the preferences of the high 
participants to embed an instance of disjointness of size~$(\delta n)^2/4$. We assume 
that the participants are
  \[
  W = \set{w_1, w_2, \ldots, w_n}, \quad M = \set{m_1, m_2, \ldots, m_n},
  \]
  where in both cases the first $\delta n / 2$ participants are high, the next 
$(1 - \delta)n / 2$ participants are mid and the remaining $n/2$ participants are low. Since the low 
and mid participants' preferences are the same for all instances, we describe those 
first. As before, the participants' preferences are symmetric in the sense that the 
men's and women's preferences are constructed analogously.
  \begin{description}
  \item[low participants] The low women's preferences over men are ``in order'': 
$m_1 \succ m_2 \succ \cdots \succ m_n$ (and symmetrically for low men, whose preference over women are ``in order''). In particular, 
each low participant prefers all high participants over all mid participants over all low participants.
  \item[mid participants] The mid participants prefer low participants over high participants over mid 
participants. Within each group, the preferences are ``in order.'' Specifically, 
the mid women have preferences
    $
    m_{n/2 + 1} \succ m_{n/2 + 2} \succ \cdots \succ m_{n} \succ m_1 \succ m_2 \succ \cdots \succ m_{n/2},
    $
    and symmetrically for the men.
  \item[high participants]  We use the preferences of each of the high participants to encode a bit 
vector of length $\delta n / 2$. Together, the men and women's preferences thus 
encode an instance of $\DISJ$ of size~$(\delta n)^2/4$. For each $w_i \in W_h$, 
we denote her bit vector $x^i_1,\ldots,x^i_{\delta n/2}$;
the preference list of $w_i$, from most-preferred to least-preferred, is:
  \begin{enumerate}
  \item men $m_j \in M_h$ such that $x_j^i = 1$;
  \item men $m \in M_l$;
  \item men $m \in M_m$;
  \item men $m_j \in M_h$ such that $x_j^i = 0$.
  \end{enumerate}
  Within each group, the preferences are once again ``in order'', i.e.\ sorted by numeric index. The men's 
preferences are constructed analogously, with each man $m_j$ encoding the bit vector $y^1_j,\ldots,y^{\delta n/2}_j$ and preferring first and foremost women $w_i\in W_h$ such that $y^i_j=1$.
  \end{description}

  \subsubsection{Stable Marriage Description}

  \begin{lemma}
    \label{lem:disj}
    Any instance of the stable marriage problem with preferences described 
above corresponding to $\DISJ(\bfx, \bfy) = 1$ has a unique stable marriage 
$\marriage_1$ given by (see the left side of Figure~\ref{fig:marriages})
\begin{align*}
    \marriage_1 =\, &\set{(m_i, w_{i + n/2}) \st i = 1, 2, \ldots, n/2}\\
    &\cup \set{(m_{i + n/2}, w_i) \st i = 1, 2, \ldots, n/2}.
\end{align*}

  \begin{figure}[ht]
    \centering
    \includegraphics{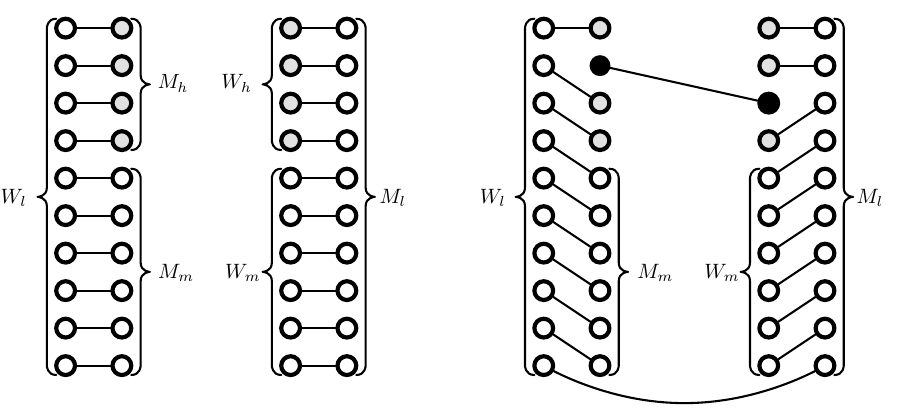}
    \caption{The (unique) stable marriages $\marriage_1$ for disjoint (left) and $\marriage_0$ for uniquely-intersecting (right) instances of the preferences described in Section \ref{sec:pref-description}. \label{fig:marriages}}
  \end{figure}

  \end{lemma}

  \begin{proof}
Let $\marriage$ be a stable marriage; we will show that $\marriage=\marriage_1$.
    We first argue that every high and mid participant is married to a low 
participant in $\marriage$. Suppose to the contrary that some $w = w_i$ for $i \leq n/2$ is 
married to some $m = m_j$ with $j \leq n/2$ in $\marriage$. By 
the  definition of the preferences and the assumption that $\DISJ(\bfx,\bfy) = 1$, 
at least one of $w$ and $m$ prefers every low participant over their spouse. Assume 
without loss of generality that $w$ prefers all $m' = m_{j'}$ with 
$j' > n / 2$ over $m$. That is, $w$ prefers all low men over her spouse $m$. 
Since $w$ is married to a medium or high man, there must be some low man $m'$ 
that is married to a low woman $w'$. But $m'$ prefers all high and medium women 
over $w'$. In particular, he prefers $w$ over $w'$. Therefore, $(w,m')$ is a 
blocking pair, so $\marriage$ is not stable. Thus any stable marriage must 
marry low participants to mid or high participants and \emph{vice versa}.

    Now we argue that if $(w_i,m_{j + n/2}) \in \marriage$, then we must
have $i = j$. The argument for pairs $(w_{i + n/2}, m_j)$ is identical. 
Suppose that $(w_i, m_{j + n/2}) \in \marriage$ with $i < j$. Then there is some 
$j' < j$ such that $m' = m_{j' + n/2}$ is married to $w' = w_{i'}$ with 
$i' > i$. But then $(w_i, m')$ mutually prefer each other, contradicting the 
stability of $\marriage$. We arrive at a similar contradiction if $i > j$, 
hence we must have $i = j$, as desired.
  \end{proof}

  \begin{lemma}
    \label{lem:int}
    Suppose we have a stable marriage instance with preferences described above 
corresponding to $\DISJ(\bfx, \bfy) = 0$, with $\bfx$ and $\bfy$ uniquely 
intersecting. Let $x^\alpha_\beta = y^\alpha_\beta = 1$ be the uniquely-intersecting
entry of $\bfx,\bfy$. In this case, there exists a unique stable marriage 
$\marriage_0$ given by (see the right side of Figure \ref{fig:marriages})
  \begin{align*}
    \marriage_0 =\, & \set{(w_{\alpha}, m_{\beta})} \cup \set{(w_i, m_{i + n / 2}) \st i < \alpha}\\
    &\cup \set{(w_{i + n/2}, m_i) \st i < \beta}\\ 
    &\cup \set{(w_i, m_{i + n / 2 - 1}), \alpha < i \leq n / 2}\\
    &\cup \set{(w_{i + n/2 - 1}, m_i), \beta < i \leq n/2} \cup \set{(w_n, m_n)}.
    \end{align*}
  \end{lemma}
  \begin{proof}
    We first argue that $(w_\alpha, m_\beta) \in \marriage$ for any 
stable marriage $\marriage$ for the preferences described above. Since $\marriage$ 
is stable, if $(m_\alpha, w_\beta) \notin \marriage$, then at least one of $w_\alpha$ and 
$m_\beta$, say $w_\alpha$, must be married to someone she prefers over 
$m_\beta$. From $w_\alpha$'s preferences, this implies that $(w_\alpha, m) \in 
\marriage$ for some $m = m_j$ with $j < \beta$ for which $x^\alpha_j = 1$. Since 
the instance of $\DISJ$ is uniquely intersecting, we must have 
$y^\alpha_j = 0$. Thus $m$ prefers all low women over $w_\alpha$. Since at most 
$n/2 - 1$ medium and high men are married to low women (indeed $m$ is a high 
man married to a high woman) and there are $n/2$ low women, some low woman $w$ is 
married to a low man. But then $w$ and $m$ mutually prefer each other, 
hence forming a blocking pair. Thus, we must have 
$(w_\alpha, m_\beta) \in \marriage$.

    The remainder of the proof of the lemma is analogous to the proof of Lemma 
\ref{lem:disj} if we remove $w_\alpha$ and $m_\beta$ from all the 
participants' preferences.
  \end{proof}

  \begin{lemma}
    \label{lem:dist}
    The marriages $\marriage_0$ and $\marriage_1$ from the previous two lemmas 
satisfy
    $
    d(\marriage_0, \marriage_1) \geq (1 - \delta) n.
    $
  \end{lemma}
  \begin{proof}
    This follows from the following two observations:
    \begin{enumerate}
    \item All mid women and men $M_m \cup W_m$ have different spouses in 
$\marriage_0$ and $\marriage_1$.
    \item No mid women are married to mid men in either $\marriage_0$ or 
$\marriage_1$.
    \end{enumerate}
    From these facts, we can conclude that 
    $
    d(\marriage_0, \marriage_1) = n - \abs{\marriage_0 \cap \marriage_1} \geq \abs{W_m} + \abs{M_m} = (1 - \delta) n.
    $
  \end{proof}

  \subsection{Derivation of Main results}
  \label{sec:proofs}

  In this section we use the construction of Section~\ref{sec:embed} to prove Theorems~\ref{thm:asm} and~\ref{thm:related} from Section~\ref{sec:summary}. 

  \begin{proof}[Proof of Theorem \ref{thm:asm}]
    Suppose that $\Pi$ is a randomized communication protocol (between Alice and Bob) that outputs a $(1 - \e)$-stable marriage $\marriage$
using $B$ bits of communication. As $\e<1/2$, there exists $\delta$ sufficiently small such 
that $\e < (1 - \delta)/2$. Suppose $\Pi$ outputs a $(1 - \e)$-stable marriage $\marriage$ for the preferences described in Section \ref{sec:pref-description}.
If $\DISJ(\bar{x},\bar{y})=1$, then by Lemma~\ref{lem:disj}, $\marriage_1$ is the unique stable marriage, so $d(\marriage,\marriage_1) \le \e n$.

Suppose $\DISJ(\bar{x},\bar{y})=0$. By Lemma~\ref{lem:int}, $\marriage_0$ is the unique stable marriage, so $d(\marriage,\marriage_0)\le\e n < (1 - \delta) n / 2$. Applying Lemma \ref{lem:dist} and the triangle inequality, we obtain
$
d(\marriage_1, \marriage) >
(1 - \delta) n / 2 >
\e n$. Thus, if $\DISJ(\bfx, \bfy) = 1$, then $d(\marriage, \marriage_1) < \e n$ and if $\DISJ(\bfx, \bfy) = 0$, then $d(\marriage, \marriage_1) > \e n$. Given~$\marriage$, Alice and Bob can compute $d(\marriage, \marriage_1)$ without communication, so they can use $\Pi$ to determine the 
value of $\DISJ(\bar{x},\bar{y})$ using $B$ bits of communication. Thus, $B = \Omega(n^2)$ by Theorem~\ref{thm:disj}, as desired. 
  \end{proof}

\begin{proof}[Proof of Theorem~\ref{thm:related}]
        For Part~(\ref{thm:related-verify}), suppose that $\Pi$ is a randomized communication protocol
that determines whether a given marriage $\marriage$ is stable or $\e$-unstable with 
respect to given preferences
using $B$ bits of 
communication. As $\e<1$, there exists $\delta$ sufficiently small such 
that $1 - \delta > \e$. Let $\marriage_1$ be the marriage defined in Lemma \ref{lem:disj}; by that lemma,
if $\DISJ(\bfx,\bfy)=1$, then $\marriage_1$ is stable (with respect to the preferences described in Section~\ref{sec:pref-description}). By Lemmas 
\ref{lem:int} and \ref{lem:dist}, if $\DISJ(\bfx,\bfy)=0$, then $\marriage_1$ is $\e$-unstable. Thus, if $\Pi$ determines whether $\marriage_1$ is stable or $\e$-unstable,
then $\Pi$ also determines the value of $\DISJ(\bfx, \bfy)$, hence 
$B = \Omega(n^2)$ by Theorem \ref{thm:disj}.

    For Part~(\ref{thm:related-ms}), suppose that $\Pi$ is a randomized communication protocol that for a given pair $(w, m)$ 
determines whether $(w, m) \in \marriage$ for some (every) stable marriage $\marriage$
using $B$ bits of communication. Set $\delta=1$. By choosing preferences as in Section~\ref{sec:pref-description} and taking $(w,m) = (w_n, m_n)$, by Lemmas~\ref{lem:disj} and 
\ref{lem:int}, $(w,m)$ is in some (equivalently every) stable marriage for the
given preferences if and only if $\DISJ(\bfx, \bfy) = 0$. Thus, once again by Theorem~\ref{thm:disj},
$B = \Omega(n^2)$.
	
    Finally, for Part~(\ref{thm:related-sp}), suppose that $\Pi$ is a randomized communication protocol that outputs $\e n$ pairs 
contained in some (every) stable marriage using $B$ bits of communication. 
Choose preferences as described in 
the Section~\ref{sec:pref-description} with some $0<\delta < \e$, say $\delta = \e / 2$.
Recall from the proof of Lemma~\ref{lem:dist} that no participants in $W_m$ and $M_m$ are ever married to one another in a stable marriage. Therefore, since $\abs{W_m} + \abs{M_m} = (1-\delta) n > (1-\e) n$ and since $\Pi$ outputs $\e n$ pairs, we have that
$\Pi$ must 
output some pair $(w,m)$ with $w \in W_m$ or $m \in M_m$. Recall from the proof of Lemma~\ref{lem:dist} that
knowing the stable spouse of any participant in $W_m$
or $M_m$ reveals the value of $\DISJ(\bfx, \bfy)$. Thus, by Theorem~\ref{thm:disj}, $B = \Omega(n^2)$.
  \end{proof}

  \section{Computing Distance to Stability}
  \label{sec:dist-to-stability}

  In this section, we describe an efficient method for computing the divorce distance to stability of a given marriage, $\marriage$. Recall that the divorce distance to stability is given by
  \[
  d(\marriage) = \min_{\marriage' \in \calM} d(\marriage, \marriage') \quad\text{where}\quad d(\marriage, \marriage') = n - \abs{\marriage \cap \marriage'}.
  \]
  In fact, our algorithm solves the following general problem, which we believe may be of independent interest: given a marriage market $\marriagemarket$ and an arbitrary marriage $\marriage$, find a stable marriage $\marriage'$ that shares the greatest number of pairs with $\marriage$. Since the set $\calM$ of all stable marriages can be exponentially large~\cite{Knuth76,IL86,GI89}, brute-force computation of $d(\marriage)$ is infeasible. Fortunately, by exploiting the structure of $\calM$, we are able to efficiently reduce the computation of $d(\marriage)$ to a max-flow/min-cut problem of size quadratic in $n$. Thus, any number of efficient algorithms may be applied to compute $d(\marriage)$.

  \subsection{The Rotation Poset and Digraph}

  Our exposition follows the work of Gusfield \cite{Gusfield87} and of Irving and Leather \cite{IL86} (see also \cite{GI89}). Let $\marriage$ be a stable marriage. A \dft{rotation $\mathbfit{\rho}$ exposed by $\mathbfit{\marriage}$} is a sequence of pairs $(w_0, m_0), (w_1, m_1), \ldots, \linebreak (w_{r-1}, m_{r-1}) \in \marriage$ such that for each $i$, $w_{i+1}$ is the first woman on $m_i$'s preference list that prefers $m_i$ to her partner $m_{i+1}$ in $\marriage$ (where addition is conducted modulo $r$). Given $\marriage$ and $\rho$, we form a marriage called the \dft{elimination of~$\mathbfit{\rho}$ from $\mathbfit{\marriage}$}, denoted  $\marriage / \rho$, which contains the pairs $(w_1, m_0), (w_2, m_1), \ldots, (w_0, m_{r-1})$, in addition to all pairs from $\marriage$ that are not part of the rotation~$\rho$. It is straightforward to verify that $\marriage / \rho$ is a stable marriage.

  Let $\marriage_0$ denote the $M$-optimal stable marriage (the marriage found by the Gale-Shapley algorithm). Irving and Leather \cite{IL86} prove that every stable marriage can be obtained from $\marriage_0$ by successively eliminating a \emph{unique} set of rotations that appear in $\marriage_0$ and subsequent stable marriages. Given a stable marriage $\marriage$, let $S_\marriage$ denote this unique set of rotations, which can be eliminated (starting at $\marriage_0$) to obtain $\marriage$. ($S_{\marriage}$ may contain some rotations that are not exposed in $\marriage_0$, but only in subsequent marriages.)

  We denote the set of all rotations exposed in one or more stable marriages in $\calM$ by $\Pi(\calM)$. We endow $\Pi(\calM)$ with a partial order $\prec$ where $\rho \prec \sigma$ if $\rho \in S_\marriage$ for every $\marriage \in \calM$ in which $\sigma$ is exposed. In other words, $\rho \prec \sigma$ if whenever $\sigma$ is eliminated during the construction of a stable marriage by elimination from $\marriage_0$, it is the case that $\rho$ has been eliminated before $\sigma$. A subset $S \subseteq \Pi(\calM)$ is (downward)~\dft{closed} if for all $\sigma \in S$ and $\rho \prec \sigma$, we have that $\rho \in S$. Irving and Leather prove the following remarkable correspondence between closed subsets of $\Pi(\calM)$ and stable marriages.

  \begin{theorem}[Irving and Leather \cite{IL86}]
    \label{thm:il86}
    The map $\marriage \mapsto S_{\marriage}$ is a bijection between~$\calM$ and the set of closed subsets of $\Pi(\calM)$. 
  \end{theorem}

  For algorithmic purposes, it is advantageous to have a sparse representation of the partial order~$\prec$ on $\Pi(\calM)$, which preserves its closed subsets. To this end, Gusfield \cite{Gusfield87} proved the following theorem. For the remainder of this section, we use the standard notation $\tilde{\Oh}$ to suppress $\log n$ factors, which we find less interesting in the context of our discussion of time complexity in this section (in contrast to the discussion of communication and query complexity in the rest of this paper).

  \begin{theorem}[Gusfield \cite{Gusfield87}]
    \label{thm:gusfield87}
    There exists a directed acyclic graph $G(\calM)$ with vertices $\Pi(\calM)$, called the \dft{rotation digraph}, whose transitive closure is the partial order $\prec$ on $\Pi(\calM)$. (That is, a path from a vertex $\rho$ to a vertex $\sigma$ exists iff $\rho\prec\sigma$.) $G(\calM)$ can be computed from the full preferences of all participants in time $\tilde{\Oh}(n^2)$. In particular, the edge and vertex sets of $G(\calM)$ both have cardinality $\Oh(n^2)$.
  \end{theorem}

  \subsection{Rotation Weights}

  In this section, we show how to assign weights to rotations $\rho \in \Pi(\calM)$ in such a way that $d(\marriage, \marriage')$ can be computed directly from $d(\marriage, \marriage_0)$ and $S_{\marriage'}$. Let $\rho$ be any rotation and $\marriage'$ a stable marriage in which $\rho$ is exposed. For any marriage $\marriage$, we define the \dft{weight} $\gamma$ of $\rho$ relative to $\marriage$ by
  \[
  \gamma_\marriage(\rho) = \abs{\marriage \cap (\marriage' / \rho)} - \abs{\marriage \cap \marriage'}.
  \]
(The absence of $\marriage'$ from the notation $\gamma_\marriage(\rho)$ becomes clear in Eq.~\ref{eq:fast-gamma} below.) That is, $\gamma_\marriage(\rho)$ is the net change, following the elimination of $\rho$ from $\marriage'$, in the number of pairs contained in the intersection of $\marriage$ and $\marriage'$. We note that $\gamma_\marriage(\rho)$ can be computed directly from $\rho$ (without being given an explicit stable marriage $\marriage'$ which exposes $\rho$). Specifically, letting $\rho' = \set{(w_1, m_0), (w_2, m_1), \ldots, (w_0, m_{r-1})}$ be the set of pairs replacing $\rho$ when eliminating $\rho$ from any stable marriage,\footnote{Note that in general, $\rho'$ is not a rotation.} we have
  \begin{equation}\label{eq:fast-gamma}
  \gamma_\marriage(\rho) = \abs{\marriage \cap \rho'} - \abs{\marriage \cap \rho}.
  \end{equation}

  \begin{lemma}
    \label{lem:weights}
    For any marriage $\marriage$ and stable marriage $\marriage' \in \calM$,
    \[
    \abs{\marriage \cap \marriage'} = \abs{\marriage \cap \marriage_0} + \sum_{\rho \in S_{\marriage'}} \gamma_\marriage(\rho),
    \]
    where $\marriage_0$ is the $M$-optimal stable marriage.
  \end{lemma}
  \begin{proof}
    We argue by induction on $\abs{S_{\marriage'}}$. If $S_{\marriage'} = \varnothing$, then $\marriage' = \marriage_0$, so the result is immediate. Suppose the claim is true for $\marriage'$ and $\sigma$ is a rotation exposed in $\marriage'$, then by the induction hypothesis and by definition of $\gamma_\marriage(\sigma)$,
    \[
    \abs{\marriage \cap (\marriage' / \sigma)} = \abs{\marriage \cap \marriage'} + \bigl(\abs{\marriage \cap (\marriage' / \sigma)} - \abs{\marriage \cap \marriage'}\bigr) = \abs{\marriage\cap\marriage_0} + \sum_{\rho \in S_{\marriage'}} \gamma_\marriage(\rho) + \gamma_\marriage(\sigma),
    \]
    which gives the desired result.
  \end{proof}

  Applying Lemma \ref{lem:weights} and Theorem \ref{thm:il86}, we obtain the following result.

  \begin{theorem}
    \label{thm:divorce-equivalence}
    Let $\marriage$ be a marriage. Then
    \[
    d(\marriage) = d(\marriage, \marriage_0) - \max_{S} \sum_{\rho \in S} \gamma_\marriage(\rho),
    \]
    where the maximum is taken over closed subsets $S \subseteq \Pi(\calM)$ (and where $\marriage_0$ is the $M$-optimal stable marriage).
  \end{theorem}
  \begin{proof}
    By Lemma \ref{lem:weights}, for any stable marriage $\marriage'$,
    \[
    d(\marriage, \marriage') = n - \abs{\marriage \cap \marriage'} = n - \abs{\marriage \cap \marriage_0} - \sum_{\rho \in S_{\marriage'}} \gamma_\marriage(\rho) = d(\marriage, \marriage_0) - \sum_{\rho \in S_{\marriage'}} \gamma_\marriage(\rho).
    \]
    By Theorem \ref{thm:il86}, $\marriage' \mapsto S_{\marriage'}$ is a bijection onto the set of closed subsets of $\Pi(\calM)$. Thus,
    \[
    d(\marriage) = \min_{\marriage' \in \calM} d(\marriage, \marriage') = d(\marriage, \marriage_0) - \max_{S} \sum_{\rho \in S} \gamma_\marriage(\rho),
    \]
    as desired.
  \end{proof}

  \subsection{Reduction to Max-Flow/Min-Cut}

We now wish to apply Theorems~\ref{thm:gusfield87} and~\ref{thm:divorce-equivalence} to efficiently calculate $d(\marriage)$.
  The divorce distance $d(\marriage, \marriage_0)$ can easily be computed in $\tilde{\Oh}(n^2)$ time by using the Gale-Shapley algorithm to compute $\marriage_0$.
 Since the partial order $\prec$ on~$\Pi(\calM)$ is the transitive closure of $G(\calM)$ (where $G(\calM)$ is the rotation digraph described in Theorem~\ref{thm:gusfield87}), the closed subsets of $\prec$ are precisely the (downward)~\dft{closed subsets} (vertex sets with no incoming edges) of $G(\calM)$, and so to compute $d(\marriage)$ it suffices to maximize $\sum_{\rho \in S} \gamma_\marriage(\rho)$ over closed subsets $S$ of $G(\calM)$.
Thus, we have reduced the problem of computing $d(\marriage)$ to finding a \dft{maximum closed subset} (i.e.\ maximum-weight vertex set with no incoming edges) in a directed acyclic graph. This problem is well studied, in particular for its applications to open-pit mining (see, e.g.\ \cite{Picard76, Hochbaum01}). For completeness and due to some differences in terminology between this paper and that of Picard~\cite{Picard76}, we briefly describe the application to our specific maximum closed subset problem of Picard's \cite{Picard76} efficient reduction of maximum closed subset to max-flow/min-cut.\footnote{Picard \cite{Picard76} and Hochbaum \cite{Hochbaum01} use the term ``closure''/``closed'' to refer to an \emph{upward}~closed subset, i.e.\ a vertex set with no \emph{outgoing} edges.
This disagrees with Irving and Leather's \cite{IL86} (and our) usage of closed to mean \emph{downward}~closed (vertex set with no \emph{incoming} edges), which appears to be more prevalent in the stable marriage literature. For consistency within this paper, when describing Picard's construction below and when phrasing Picard's theorem as Theorem~\ref{thm:st-cut}, we perform the trivial needed modifications (flipping the direction of all edges as well as the roles of the source and sink vertices) to naturally present Picard's results with respect to downward rather than upward closed subsets.}

Denote the vertex and edge sets of $G(\calM)$ by $V=\Pi(\calM)$ and $E$ respectively. Let
  \[
  V^+ = \set{\rho \in V \st \gamma_\marriage(\rho) \geq 0} \quad\text{and}\quad V^- = \set{\rho \in V \st \gamma_\marriage(\rho) < 0}.
  \]
  We add a source vertex $s$ and a sink vertex $t$ to $G(\calM)$ to form a new $st$-graph $\tilde{G} = (\tilde{V}, \tilde{E})$ where $\tilde{V} = V \cup \set{s, t}$ and
  \[
  \tilde{E} ~=~ E ~\cup~ \set{(s, \rho) \st \rho \in V^-} ~\cup~ \set{(\rho, t) \st \rho \in V^+}.
  \]
We assign nonnegative capacity $c(u, v)$ to each edge $(u, v) \in \tilde{E}$ by
  \[
  c(u, v) =
  \begin{cases}
    \infty & (u, v) \in E,\\
    \gamma_\marriage(\rho) & (u, v) = (\rho, t),\\
    -\gamma_\marriage(\rho) & (u, v) = (s, \rho).
  \end{cases}
  \]
  
  \begin{theorem}[Picard \cite{Picard76}]
    \label{thm:st-cut}
    The sink set (i.e.\ the set of vertices on the same side as the sink $t$) of a minimum $st$-cut in $\tilde{G}$ is a maximum closed subset in $G(\calM)$.
  \end{theorem}

  In light of Theorem \ref{thm:st-cut}, we can reduce the computation of $d(\marriage)$ to known efficient algorithms for max-flow/min-cut. We summarize the procedure as follows.

  \begin{algorithm}
    \caption{$\mathrm{DivorceDistance}(\marriage, P_W, P_M)$ --- compute $d(\marriage)$ with respect to preferences $P_W, P_M$.}
    \label{alg:divorce-distance}
    \begin{enumerate}
    \item Use the Gale-Shapley algorithm to compute $\marriage_0$ and compute $d(\marriage, \marriage_0) = n - \abs{\marriage \cap \marriage_0}$.
    \item Construct the rotation digraph $G(\calM)$ and the related graph $\tilde{G}$.
    \item Compute the weights in $\tilde{G}$ by computing $\gamma_\marriage(\rho)$ for each rotation $\rho \in G(\calM)$, using Eq.~\eqref{eq:fast-gamma}.
    \item Find a minimum $st$-cut in $\tilde{G}$, and let $S$ be the sink set in the cut.
    \item Return $d(\marriage, \marriage_0) - \sum_{\rho \in S} \gamma_\marriage(\rho)$.
    \end{enumerate}
  \vspace{-1em}
  \end{algorithm}

  \begin{theorem}
    \label{thm:divorce-distance}
    Given a marriage market $\marriagemarket$ and a marriage $\marriage$, running the procedure $\mathrm{DivorceDistance}(\marriage, P_W, P_M)$ computes the divorce distance to stability $d(\marriage)$ in time $\tilde{\Oh}(n^4)$.
  \end{theorem}

  We remark that since $(P_W,P_M)$ has size $\tilde{\Theta}(n^2)$, the runtime of $\mathrm{DivorceDistance}$ is nearly quadratic in the input size.

  \begin{proof}[Proof of Theorem~\ref{thm:divorce-distance}]
    The correctness of $\mathrm{DivorceDistance}(\marriage,P_W,P_M)$ follows immediately from Theorems~\ref{thm:divorce-equivalence} and~\ref{thm:st-cut}. We analyze the runtime as follows. Step 1 can be computed in time $\tilde{\Oh}(n^2)$ using the Gale-Shapley algorithm and brute force computation of $d(\marriage, \marriage_0)$. For Step 2, by Theorem \ref{thm:gusfield87}, $G(\calM)$ (and hence $\tilde{G}$) can also be computed in time $\tilde{\Oh}(n^2)$. The weights in Step 3 can be computed in linear time for each rotation, so computing all the weights can be accomplished in time~$\tilde{\Oh}(n^3)$. For Step 4, the min-cut can be computed in time $\tilde{\Oh}\bigl(|\tilde{E}||\tilde{V}|\bigr) = \tilde{\Oh}(n^4)$ using, for example, Hochbaum's algorithm \cite{Hochbaum01}. 
  \end{proof}

We remark that since Algorithm~\ref{alg:divorce-distance} finds both $\marriage_0$ and $S=S_{\marriage'}$ for a stable marriage $\marriage'$ closest to $\marriage$, it is a trivial task, which does not increase the asymptotic runtime complexity of Algorithm~\ref{alg:divorce-distance}, to also compute $\marriage'$ in addition to computing $d(\marriage)=d(\marriage,\marriage')$.

  \section{Commentary and Open Problems}
  \label{sec:commentary}

  A number of recent papers~\cite{AKL13,G14} have touched on various aspects of the amorphic question of ``how much do the preferences of the women in the Gale-Shapley algorithm affect the produced ($M$-optimal) stable marriage.'' The fact that we prove our lower-bound result in a strong two-sided communication model (and not a weaker $2n$-sided communication model or an even-weaker query model) allows our results to also be viewed in the context of this line of research. Our communication lower bounds show that a significant amount of information about the preferences of the women is indeed needed in order to deduce the $M$-optimal stable marriage, as well as for solving any of the other problems described in Corollary~\ref{cor:related}.

  One qualitative feature of of the Gale-Shapley algorithm is that a single proposal at any point can precipitate a cascade of rejections that affects a large portion of the population. Thus, it is impossible for participants to know if their current partner is their final partner until the algorithm has terminated. Our results imply that this feature is common to \emph{all} stable marriage mechanisms that dynamically refine a marriage and converge to a stable marriage. Indeed, consider any stable marriage algorithm and arbitrarily divide it into a ``first'' stage and a ``second'' stage. A consequence of Corollary~\ref{cor:asm} along with our novel definition of divorce distance is that if the query complexity of the first stage is significantly lower than that of querying the entire input, then after the first stage a large fraction of the participants may not yet be married to their final spouses.

    In many real-world marriage markets, centralized clearinghouses are employed to prevent undesirable outcomes~\cite{Roth1984, Roth1991}. Specifically, these clearinghouses were implemented to avoid ``unraveling'' --- wherein participants are incentivized to match extremely early --- as well as instability. While unraveling is undesirable in its own right, the early binding commitments made in an unraveling market have been shown to have adverse effects~\cite{Roth1994, Li1998, Niederle2003, Frechette2007}, presumably because the early matches are necessarily made with incomplete knowledge about the market. A consequence of Corollary~\ref{cor:related} (specifically, of Part~(\ref{cor:related-sp})) is that any marriage mechanism that allows even a small fraction of participants to match in early binding commitments (i.e.\ before essentially all of the preferences are queried) cannot generally produce a stable marriage. Thus the empirical phenomenon of instability in decentralized markets where early-accepted proposals are binding, which was observed in~\cite{Roth1984, Roth1991}, is not merely a feature of the particular marriage mechanisms that arose in practice, but is a general theoretical feature inherent in the stable marriage problem.
  
It is interesting to compare the $\Omega(n^2)$ lower bound that is proved in this paper for the communication complexity of finding an approximately stable marriage to known complexity bounds for the problem of finding an approximately maximum-weight matching in a bipartite graph. Even though these two problems seem similar, the latter can be solved with $\Oh(n\log n)$ communication~\cite{DNO14}, i.e.\ with significantly less communication than many variants of the former. It is worthwhile to compare this surprising dissimilarity between these problems with a qualitatively similar message that emerges from a significantly different, recent line of work~\cite{Li2015,AG2016}, which shows that finding a Pareto-efficient perfect matching requires considerably less strategic reasoning than finding a stable marriage.

  The classic Gale-Shapley algorithm~\cite{GS62} terminates after $\Oh(n^2)$ steps, and each step consists of a message of $\Oh(\log n)$ 
bits. Thus, the Gale-Shapley algorithm provides a communication upper bound of
$\Oh(n^2 \log n)$ for the problem of finding a stable marriage. As mentioned in the introduction, our
Theorem~\ref{thm:asm} matches this up to a logarithmic factor, but it is not
immediately clear how to close this gap.

\begin{open}
  What is the
communication complexity of finding a stable marriage? 
\end{open}

  Our definition of $(1 - \e)$-stability is nonstandard. A more common notion
of approximate stability is that a marriage induce few (say, at most $\e n^2$) 
blocking pairs (see~\cite{EH08}). As shown in Propositions~\ref{prop:blocking-pairs}, the blocking-pairs notion
of approximate stability is strictly coarser than ours. It is therefore natural to ask if the 
$\Omega(n^2)$ communication lower bound of Theorem~\ref{thm:asm} holds as well with respect to
the more challenging concept of blocking-pairs approximate stability.

\begin{open}
  Is there a protocol $\Pi$ that computes a marriage with at most $\e n^2$
blocking pairs using $o(n^2)$ communication? 
\end{open}

  Recently, Ostrovsky and Rosenbaum~\cite{OR14} showed that it is possible to find a marriage with $\e n^2$ blocking pairs for arbitrary $\e > 0$ using $\Oh(1)$  communication \emph{rounds} for a distributed model of computation. While their result does not imply anything nontrivial about the total communication, we believe their techniques may be relevant for finding $o(n^2)$ communication protocols for blocking-pairs approximate stability (if such protocols exist).
Interestingly, an analogue of Theorem~\ref{thm:related}(\ref{thm:related-verify}) (or of Corollary~\ref{cor:related}(\ref{cor:related-verify})) does not hold 
for blocking-pairs approximate stability.

\begin{theorem}\label{verify-few-blocking-pairs}
  For every $\e \geq \delta > 0$, there exists a randomized communication 
protocol $\Pi$ that determines whether a given marriage $\marriage$ induces at 
least $\e n(n-1)$ blocking pairs or at most ${(\e - \delta) n(n-1)}$ blocking pairs 
using $\Oh(\log n)$ communication. In particular, $\Pi$ determines whether $\marriage$ 
is stable or has $\e n(n-1)$ blocking pairs using $\Oh(\log n)$ communication.
\end{theorem}
\begin{proof}[Proof sketch.]
  Choose an unmarried pair $(w, m)$ uniformly at random from $W \times M$. If $m$ prefers
$w$ over his spouse in~$\marriage$, the men query the women to see if $w$ also
prefers $m$ over her spouse in $\marriage$ using $\Oh(\log n)$ communication.
The probability that $(w, m)$ is a blocking pair is precisely $\e'$, where
$\e'$ is the fraction of blocking pairs in $\marriage$.
Repeat this procedure to estimate $\e'$ to any desired accuracy in a bounded 
number of steps depending only on the desired accuracy.
\end{proof}

Theorem~\ref{thm:related}(\ref{thm:related-sp}) shows that any protocol that produces a constant 
fraction of pairs in a stable marriage (regardless of \emph{which} pairs are
found) requires $\Omega(n^2)$ communication. It would be interesting to 
improve this result (or find an efficient protocol) for finding even a single
pair that appears in a stable marriage.

\begin{open}\label{open:single-pair}
  What is the communication complexity of finding a single pair $(w, m)$ that
appears in some/every stable marriage?
\end{open}

Finally, we notice that in contrast to Parts~(\ref{thm:related-verify}) and~(\ref{thm:related-sp}) of Theorem~\ref{thm:related}, our statement of Theorem~\ref{thm:asm} requires that $\e<1/2$. It is natural to ask what can be obtained regarding other values of $\e$.

\begin{open}
Fix $\frac{1}{2}\le\e<1$. What is the communication complexity of finding a \mbox{$(1-\e)$-stable} marriage?
\end{open}

\section*{Acknowledgements}
Yannai Gonczarowski is supported by the Adams Fellowship Program of the Israel Academy of Sciences and Humanities. The work of Yannai Gonczarowski was supported in part by the European Research Council under the European Community's Seventh Framework
Programme (FP7/2007-2013) / ERC grant agreement no.\ [249159].

The work of Noam Nisan was supported in part by ISF grants 230/10 and 1435/14 administered by the Israeli Academy of Sciences, and by
Israel-USA Bi-national Science Foundation (BSF) grant number 2014389.

The work of Rafail Ostrovsky was supported in part by NSF grants 09165174, 1065276, 1118126, 1136174 and 1619348; US-Israel BSF grant
2012366, OKAWA Foundation Research Award, IBM Faculty Research Award, Xerox Faculty Research Award, B.\ John Garrick Foundation Award, Teradata Research Award, and Lockheed-Martin Corporation Research Award. This material is based upon work supported
in part by DARPA SafeWare program.
The views expressed are those of the author and do not reflect the official policy or position of the Department of Defense or the U.S.\ Government.

We would like to thank two anonymous referees for many helpful comments.

\bibliographystyle{abbrv}
\bibliography{asm-cc-arxiv}{}

\begin{thebibliography}{10}

\bibitem{AG2016}
I.~Ashlagi and Y.~A. Gonczarowski.
\newblock Stable matching mechanisms are not obviously strategy-proof.
\newblock {\em Journal of Economic Theory}, 177:405--425, 2018.

\bibitem{AKL13}
I.~Ashlagi, Y.~Kanoria, and J.~D. Leshno.
\newblock Unbalanced random matching markets: The stark effect of competition.
\newblock {\em Journal of Political Economy}, 125(1):69--98, 2017.

\bibitem{BeiChenZhang13}
X.~Bei, N.~Chen, and S.~Zhang.
\newblock On the complexity of trial and error.
\newblock In {\em Proceedings of the 44th Annual ACM Symposium on Theory of
  Computing (STOC)}, pages 31--40, 2013.

\bibitem{Bhatnagar2008}
N.~Bhatnagar, S.~Greenberg, and D.~Randall.
\newblock Sampling stable marriages: Why spouse-swapping won't work.
\newblock In {\em Proceedings of the 19th Annual ACM-SIAM Symposium on Discrete
  Algorithms (SODA)}, pages 1223--1232, 2008.

\bibitem{Bogomolnaia2007}
A.~Bogomolnaia and J.-F. Laslier.
\newblock Euclidean preferences.
\newblock {\em Journal of Mathematical Economics}, 43(2):87--98, 2007.

\bibitem{CL10}
J.-H. Chou and C.-J. Lu.
\newblock Communication requirements for stable marriages.
\newblock In {\em Proceedings of the 7th International Conference on Algorithms
  and Complexity (CIAC)}, pages 371--382, 2010.

\bibitem{DNO14}
S.~Dobzinski, N.~Nisan, and S.~Oren.
\newblock Economic efficiency requires interaction.
\newblock In {\em Proceedings of the 46th Annual ACM Symposium on Theory of
  Computing (STOC)}, pages 233--242, 2014.
\newblock Full version available at arXiv:1311.4721.

\bibitem{DF81}
L.~E. Dubins and D.~A. Freedman.
\newblock Machiavelli and the {G}ale-{S}hapley algorithm.
\newblock {\em American Mathematical Monthly}, 88(7):485--494, 1981.

\bibitem{EH08}
K.~Eriksson and O.~H{\"a}ggstr{\"o}m.
\newblock Instability of matchings in decentralized markets with various
  preference structures.
\newblock {\em International Journal of Game Theory}, 36(3):409--420, March
  2008.

\bibitem{Frechette2007}
G.~R. Fr\'{e}chette, A.~E. Roth, and M.~U. \"{U}nver.
\newblock Unraveling yields inefficient matchings: evidence from post-season
  college football bowls.
\newblock {\em The RAND Journal of Economics}, 38(4):967--982, 2007.

\bibitem{GS62}
D.~Gale and L.~S. Shapley.
\newblock College admissions and the stability of marriage.
\newblock {\em The American Mathematical Monthly}, 69(1):9--15, 1962.

\bibitem{G14}
Y.~A. Gonczarowski.
\newblock Manipulation of stable matchings using minimal blacklists.
\newblock In {\em Proceedings of the 15th ACM Conference on Economics and
  Computation (EC)}, page 449, 2014.

\bibitem{GF13}
Y.~A. Gonczarowski and E.~Friedgut.
\newblock Sisterhood in the {G}ale-{S}hapley matching algorithm.
\newblock {\em The Electronic Journal of Combinatorics}, 20(2):\#P12 (18pp),
  2013.

\bibitem{Gusfield87}
D.~Gusfield.
\newblock Three fast algorithms for four problems in stable marriage.
\newblock {\em SIAM Journal on Computing}, 16(1):111--128, 1987.

\bibitem{GI89}
D.~Gusfield and R.~W. Irving.
\newblock {\em The Stable Marriage Problem: Structure and Algorithms},
  volume~54.
\newblock MIT Press, 1989.

\bibitem{Hochbaum01}
D.~S. Hochbaum.
\newblock A new---old algorithm for minimum-cut and maximum-flow in closure
  graphs.
\newblock {\em Networks}, 37(4):171--193, 2001.

\bibitem{IL86}
R.~W. Irving and P.~Leather.
\newblock The complexity of counting stable marriages.
\newblock {\em SIAM Journal on Computing}, 15(3):655--667, 1986.

\bibitem{KS92}
B.~Kalyanasundaram and G.~Schintger.
\newblock The probabilistic communication complexity of set intersection.
\newblock {\em SIAM Journal on Discrete Mathematics}, 5(4):545--557, 1992.

\bibitem{Knuth76}
D.~E. Knuth.
\newblock {\em Marriage stables et leurs relations avec d'autres probl\`{e}mes
  combinatoires}.
\newblock Les Presses de l'Universit\'{e} de Montr\'{e}al, 1976.

\bibitem{KN97}
E.~Kushilevitz and N.~Nisan.
\newblock {\em Communication Complexity}.
\newblock Cambridge University Press, 1997.

\bibitem{Li1998}
H.~Li and S.~Rosen.
\newblock Unraveling in matching markets.
\newblock {\em American Economic Review}, 88(3):371--387, 1998.

\bibitem{Li2015}
S.~Li.
\newblock Obviously strategy-proof mechanisms.
\newblock {\em American Economic Review}, 107(11):3257--3287, 2017.

\bibitem{MW71}
D.~G. McVitie and L.~B. Wilson.
\newblock The stable marriage problem.
\newblock {\em Communications of the ACM}, 14(7):486--490, 1971.

\bibitem{HN90}
C.~Ng and D.~Hirschberg.
\newblock Lower bounds for the stable marriage problem and its variants.
\newblock {\em SIAM Journal on Computing}, 19(1):71--77, 1990.

\bibitem{Niederle2003}
M.~Niederle and A.~E. Roth.
\newblock Unraveling reduces mobility in a labor market: Gastroenterology with
  and without a centralized match.
\newblock {\em Journal of Political Economy}, 111(6):1342--1352, 2003.

\bibitem{OR14}
R.~Ostrovsky and W.~Rosenbaum.
\newblock Fast distributed almost stable matchings.
\newblock In {\em Proceedings of the 34th ACM Symposium on Principles of
  Distributed Computing (PODC)}, pages 101--108, 2015.

\bibitem{Picard76}
J.-C. Picard.
\newblock Maximal closure of a graph and applications to combinatorial
  problems.
\newblock {\em Management Science}, 22(11):1268--1272, 1976.

\bibitem{Razborov92}
A.~A. Razborov.
\newblock On the distributional complexity of disjointness.
\newblock {\em Theoretical Computer Science}, 106, 1992.

\bibitem{Roth1984}
A.~E. Roth.
\newblock The evolution of the labor market for medical interns and residents:
  A case study in game theory.
\newblock {\em Journal of Political Economy}, 92(6):991--1016, 1984.

\bibitem{Roth86}
A.~E. Roth.
\newblock On the allocation of residents to rural hospitals: A general property
  of two-sided matching markets.
\newblock {\em Econometrica}, 54(4):425--427, 1986.

\bibitem{Roth1991}
A.~E. Roth.
\newblock A natural experiment in the organization of entry-level labor
  markets: Regional markets for new physicians and surgeons in the {U}nited
  {K}ingdom.
\newblock {\em American Economic Review}, 81(3):415--440, 1991.

\bibitem{Roth1994}
A.~E. Roth and X.~Xing.
\newblock Jumping the gun: Imperfections and institutions related to the timing
  of market transactions.
\newblock {\em American Economic Review}, 84(4):992--1044, 1994.

\bibitem{Segal03}
I.~Segal.
\newblock The communication requirements of social choice rules and supporting
  budget sets.
\newblock {\em Journal of Economic Theory}, 136:341--378, 2007.

\bibitem{Unver05}
M.~U. {\"U}nver.
\newblock On the survival of some unstable two-sided matching mechanisms.
\newblock {\em International Journal of Game Theory}, 33(2):239--254, 2005.

\bibitem{Wilson72}
L.~B. Wilson.
\newblock An analysis of the stable marriage assignment algorithm.
\newblock {\em BIT Numerical Mathematics}, 12(4):569--575, 1972.

\bibitem{Yao79}
A.~C.-C. Yao.
\newblock Some complexity questions related to distributive computing
  (preliminary report).
\newblock In {\em Proceedings of the 11th Annual ACM Symposium on Theory of
  Computing (STOC)}, pages 209--213, 1979.

\end{thebibliography}

\appendix


\section{Asymptotic Notation}
\label{sec:asymptotic-notation}

Throughout this paper, we use standard computer-science asymptotic notation to describe the order of growth of single-dimensional functions of natural numbers. For example, for a positive function~$f$, we write $f(n)=\Oh\bigl(g(n)\bigr)$ where $g$ is also a positive function (usually, one simple to write down), if there exist positive numbers $\overline{M}$ and $N$ such that $f(n)\le\overline{M}\cdot g(n)$ for all $n>N$. (Equivalently and succinctly, we write $f=\Oh\bigl(g(n)\bigr)$ if $\lim\sup_{n\rightarrow\infty}\frac{f(n)}{g(n)}<\infty$.) Intuitively, one may find it helpful to read the notation $f=\Oh\bigl(g(n)\bigr)$ as $f\in\Oh(g)$, where $\Oh(g)$ is the class of functions whose order of growth (as $n$ grows large) is at most that of $g$. This notation allows for the simplification of the exposition of many results, where the order of magnitude of the result serves as the main message. For example, for a highly complex function $f$ that satisfies $f=\Oh\bigl(n^2\bigr)$ (that is, does not grow any faster than $g$, where $g(n)=n^2$ for all $n$), instead of specifying $f$ and saying that some algorithm takes at most $f$ steps, one could concisely say that this algorithm takes at most $\Oh(n^2)$ steps, without the need to explicitly write down the complex function $f$. The following table summarizes this and other similar standard ``order of growth'' notation used throughout this paper.

\begin{center}
\begin{tabular}{llll}
Notation & Explicit Definition & Succinct Definition & Informal Meaning \\
\hhline{====}\\[-.75em]
$f=\Oh\bigl(g(n)\bigr)$ & \pbox{4cm}{\small $\exists \overline{M},N: \forall n>N:$ \\ $f(n)\le \overline{M}\cdot g(n)$} & $\lim\sup_{n\rightarrow\infty}\frac{f(n)}{g(n)}<\infty$ & $f$ grows at most as fast as $g$ \\[.75em]
\hline\\[-.75em]
$f=\Omega\bigl(g(n)\bigr)$ & \pbox{4cm}{\small $\exists\underline{M}>0,N: \forall n>N:$ \\ $f(n)\ge \underline{M}\cdot g(n)$} & $\lim\inf_{n\rightarrow\infty}\frac{f(n)}{g(n)}>0$ & $f$ grows at least as fast as $g$ \\[.75em]
\hline\\[-.75em]
$f=\Theta\bigl(g(n)\bigr)$ & \pbox{4cm}{\small $\exists\underline{M}>0,\overline{M},N:\forall n>N:$ \\ $ \underline{M}\cdot g(n) \le f(n)\le \overline{M}\cdot g(n)$} & ~~~~\pbox{3cm}{$f=\Oh\bigl(g(n)\bigr)$ \& \\ $f=\Omega\bigl(g(n)\bigr)$} & $f$ grows as fast as $g$ \\[.75em]
\hline\\[-.75em]
$f=o\bigl(g(n)\bigr)$ & \pbox{4cm}{\small $\forall\overline{M}>0:\exists N: \forall n>N:$ \\ $f(n)\le\overline{M}\cdot g(n)$} & $\lim\sup_{n\rightarrow\infty}\frac{f(n)}{g(n)}=0$ & $f$ grows slower than $g$ \\[.75em]
\hline\\[-.75em]
$f=\omega\bigl(g(n)\bigr)$ & \pbox{4cm}{\small $\forall\underline{M}:\exists N: \forall n>N:$ \\ $f(n)\ge\underline{M}\cdot g(n)$} & $\lim\inf_{n\rightarrow\infty}\frac{f(n)}{g(n)}=\infty$ & $f$ grows faster than $g$
\end{tabular}\vspace{0.5em}
\end{center}

\section{Embedding Arbitrary Preferences into Complete Preferences}
\label{sec:arbitrary-complete}

This section contains the remaining technical details needed to complete the direct proof that is given in Section~\ref{sec:smlb} of our lower bound on the communication complexity of finding an exactly stable marriage.

\begin{Definition}[Submarriage]
Let $W'$ and $M'$ be disjoint sets. A marriage $\marriage$, between a subset~$W$
of $W'$ and a subset $M$ of $M'$, is said to be a \dft{submarriage} of a 
marriage $\marriage'$ between $W'$ and $M'$, if for every $w \in W$ and $m \in M$, 
we have $\marriage'(w)=m$ iff $\marriage(w)=m$.
\end{Definition}

\begin{lemma}[Finding a Stable Marriage (Arbitrary Preferences) $\hookrightarrow$ Finding a Stable Marriage (Full Preferences)]
\label{findstabbl-in-findstab}
Let $n \in \mathbb{N}$, and let $W$, $W'$, $M$ and~$M'$ be pairwise-disjoint 
sets, each of cardinality $n$. There exist functions 
$\prefs{W\cup W'}:\allprefs(W,M)\rightarrow\mbox{$\fullprefs(W\cup W',M\cup M')$}$ 
and 
$\prefs{M\cup M'}:\allprefs(M,W)\rightarrow\mbox{$\fullprefs(M\cup M',W\cup W')$}$ 
such that for every $\prefs{W} \in \allprefs(W,M)$ and 
$\prefs{M} \in \allprefs(M,W)$, and for every (possibly imperfect) marriage 
$\marriage$ between $W$ and $M$, the following are equivalent.
\begin{itemize}
\item $\marriage$ is stable with respect to $\prefs{W}$ and $\prefs{M}$.

\item $\marriage$ is a submarriage of some marriage between $W\cup W'$ and 
$M\cup M'$ that is stable with respect to $\prefs{W\cup W'}(\prefs{W})$ and 
$\prefs{M\cup M'}(\prefs{M})$.
\end{itemize}
\end{lemma}

\begin{proof}[Proof.\footnotemark]\footnotetext{Our construction in this proof 
is essentially a one-to-one version of the many-to-many construction given in 
Corollary~31 of \cite{GF13}.}
Denote $W\!=\!\{w_1,\ldots,w_n\}$, $M\!=\!\{m_1,\ldots,m_n\}$, $W'\!=\!\{w'_1,\ldots,w'_n\}$, 
and $M'\!=\!\{m'_1,\ldots,m'_n\}$.

To define $\prefs{W\cup W'}(\prefs{W})$, for every~$i$ we define the 
preference list of $w_i$ to consist of her preference list in $\prefs{W}$ (in 
the same order), followed by $m'_i$, followed by all other men in arbitrary 
order; we define the preference list of $w'_i$ to consist of $m_i$, followed 
by all other men in arbitrary order. Similarly, to define 
$\prefs{M\cup M'}(\prefs{M})$, for every $j$ we define the preference list of 
$m_j$ to consist of his preference list in $\prefs{M}$ (in the same order), 
followed by $w'_j$, followed by all other women in arbitrary order; we define 
the preference list of $m'_j$ to consist of $w_j$, followed by all other women 
in arbitrary order.

It is straightforward to verify that the lemma 
holds with respect to these definitions of $\prefs{W\cup W'}$ and 
$\prefs{M\cup M'}$; the details are left to the reader.
\end{proof}

\begin{Remark}
It is straightforward to embed the problem of finding a stable marriage \wrt\ full preference lists in that of finding a stable marriage \wrt\ arbitrary preference lists, as the former is a special case of the latter.
\end{Remark}


\section{Determining the Marital Status of\texorpdfstring{\\}{ }a Given Couple or Participant}\label{isaremarried}

In this appendix, we give an alternate proof of Part~(\ref{thm:related-ms}) of Theorem~\ref{thm:related} that uses the construction of Section~\ref{sec:smlb}.
We prove Theorem~\ref{thm:related}(\ref{thm:related-ms}) once again using Theorem \ref{thm:disj}, by embedding disjointness (without assuming unique intersection) in both problems.
We embed disjointness via an intermediate problem of determining whether a given participant is single (i.e.\ not married to anyone) in some stable marriage, given profiles of arbitrary (i.e.\ not necessarily full) preference lists.\footnote{By Theorem \ref{thm:rural-hospitals} (in conjunction with \ref{thm:gs62}), this is equivalent to whether this participant is single in \emph{every} stable marriage.} We therefore obtain the same lower bounds for this natural problem as well.

\begin{lemma}[Disjointness $\hookrightarrow$ Is Participant Single?]
\label{disj-in-issingle}
Let $n \in \mathbb{N}$, let $W$ and $M$ be disjoint sets \sut\ $|W|=|M|=2n$, and let $p\in W \cup M$. There exist functions
$\prefs{W}:\{0,1\}^{\distinctpairs{[n]}}\rightarrow \allprefs(W,M)$ and $\prefs{M}:\{0,1\}^{\distinctpairs{[n]}}\rightarrow \allprefs(M,W)$ \sut\ for every $\bfx=(x^i_j)_{(i,j)\in\distinctpairs{[n]}}\in\{0,1\}^{\distinctpairs{[n]}}$ and $\bfy=(y^i_j)_{(i,j)\in\distinctpairs{[n]}}\in\{0,1\}^{\distinctpairs{[n]}}$, the following are equivalent.
\begin{itemize}
\item
$p$ is single in some stable marriage \wrt\ $\prefs{W}(\bfx)$ and $\prefs{M}(\bfy)$.
\item
$\DISJ(\bfx, \bfy) \neq 0$.
\end{itemize}
\end{lemma}

\begin{proof}
Assume \wwlog\ that $p\in W$ and denote $w=p$. Denote $W=\{w_1,\ldots,w_n,\allowbreak w,w'_2,w'_3,\ldots,w'_n\}$ and $M=\{m_1,\ldots,m_n,m'_1,\ldots,m'_n\}$.

To define $\prefs{W}(\bfx)$, for every $i$
we define the preference list of~$w_i$ to consist of all $m_j$ \sut~$x^i_j=1$, in arbitrary order (say, sorted by $j$), followed by $m'_i$ (with all other men absent). We define the preference list of $w$ to consist of all $m'_j$, in arbitrary order (say, sorted by $j$), with all other men absent. We define the preference list of every $w'_i$ to be empty (these women can be ignored, and are defined purely for aesthetic reasons --- so that $W$ and $M$ be of equal cardinality).
To define~$\prefs{M}(\bfy)$, for every $j$
we define the preference list of $m_j$ to consist of all $w_i$ \sut\ $y^i_j=1$, in arbitrary order (say, sorted by~$i$), with all other women absent. For every $j$
we define the preference list of $m'_j$ to consist of~$w_j$, followed by $w$ (with all other women absent).

Let $\marriageid'$ be the marriage in which $w_i$ is married to $m'_i$ for every $i$,
and in which all other participants are single. We first show that $\DISJ(\bfx, \bfy) \neq 0$ iff $\marriageid'$ is stable, and then show that $\marriageid'$ is stable iff $w=p$ is single in some stable marriage; we commence with the former.

We begin by noting that every participant that is married in $\marriageid'$ is married to someone on their preference list; therefore, $\marriageid'$ is stable iff no pair would rather deviate. Obviously, no $w'_i$ would rather deviate with anyone. Furthermore, while $w$ would rather deviate with any $m'_j$, these are all married to their top choices, and so none of them would deviate with~$w$. Since for every $i$,
the preference list of $w_i$ consists of $m'_i$ and of a subset of $\{m_j\}_{j\ne i}$, we therefore have that $\marriageid'$ is unstable iff
there exists $(i,j) \in \distinctpairs{[n]}$ \sut\ both $m_j \succ_{w_i} m'_i$ and $w_i$ is on the preference list of $m_j$.
Similarly to the proof of Lemma \ref{disj-in-findstabbl}, this holds precisely if there exists $(i,j) \in \distinctpairs{[n]}$ \sut\ $x^i_j=1$ and $y^i_j=1$, which holds iff $\DISJ(\bfx, \bfy) = 0$.

We complete the proof by showing that $\marriageid'$ is stable iff $w=p$ is single in some stable marriage. The first direction follows immediately from the fact
that $w$ is single in $\marriageid'$. For the second direction, assume that there exists a stable marriage $\marriage$ in which $w$ is single. By stability of $\marriage$ and since all men on the preference list of $w$ have $w$ on their preference list, all such men are married in $\marriage$ and prefer their spouses over $w$. Therefore, for every~$j$,
we have that $m'_j$ is married to $w_j$ in $\marriage$. By stability of $\marriage$, every $w'_i$ is single in $\marriage$. As $\marriage$ and $\marriageid'$ coincide on all women, we have that $\marriage=\marriageid'$. Therefore, $\marriageid'=\marriage$ is stable and the proof is complete.
\end{proof}

\begin{corollary}[Complexity of Determining the Marital Status of a Given Participant]
\label{issingle-cc-complexity}
The lower bounds of Theorem~\ref{thm:related} and Corollary~\ref{cor:related} apply also to the problem of determining whether a given participant $p\in W\cup M$ is single in some (equivalently, in every) stable marriage,
where $\prefs{W}\in\allprefs(W,M)$ and $\prefs{M}\in\allprefs(M, W)$ are arbitrary (i.e.\ not necessarily full) preference lists.
\end{corollary}

\begin{lemma}[Is Participant Single? $\hookrightarrow$ Is Couple Sometimes/Always Married?]
\label{issingle-in-aremarried}
Let $n \in \mathbb{N}$, and let $W$, $W'$, $M$ and~$M'$ be pairwise-disjoint sets, each of cardinality $n$; let $w\in W$ and $m'\in M'$.
There exist functions $\prefs{W\cup W'}:\allprefs(W,M)\rightarrow\mbox{$\fullprefs(W\cup W',M\cup M')$}$ and $\prefs{M\cup M'}:\allprefs(M,W)\rightarrow\mbox{$\fullprefs(M\cup M',W\cup W')$}$ \sut\ for every
$\prefs{W} \in \allprefs(W,M)$ and $\prefs{M} \in \allprefs(M,W)$, the following are equivalent.
\begin{itemize}
\item
$w$ is single in some marriage between $W$ and $M$ that is stable \wrt\ $\prefs{W}$ and~$\prefs{M}$.
\item
$w$ and $m'$ are married in {\bfseries some} marriage between $W\cup W'$ and $M\cup M'$ that is stable \wrt\ $\prefs{W\cup W'}(\prefs{W})$ and $\prefs{M\cup M'}(\prefs{M})$.
\item
$w$ and $m'$ are married in {\bfseries every} marriage between $W\cup W'$ and $M\cup M'$ that is stable \wrt\ $\prefs{W\cup W'}(\prefs{W})$ and $\prefs{M\cup M'}(\prefs{M})$.
\end{itemize}
\end{lemma}

\begin{proof}
The proof is similar to that of Lemma \ref{findstabbl-in-findstab}.
Denote $W=\{w_1=w,w_2,\ldots,w_n\}$, $M=\{m_1,\ldots,m_n\}$, $W'=\{w'_1,\ldots,w'_n\}$, and $M'=\{m'_1=m',m'_2,\ldots,m'_n\}$.

To define $\prefs{W\cup W'}(\prefs{W})$, for every~$i$
we define the preference list of $w_i$ to consist of her preference list in $\prefs{W}$ (in the same order),
followed by $m'_i$, followed by all other men in arbitrary order; we define the preference list of $w'_i$ to consist of $m_i$, followed by all other men in arbitrary order.
Similarly, to define $\prefs{M\cup M'}(\prefs{M})$, for every $j$
we define the preference list of $m_j$ to consist of his preference list in $\prefs{M}$ (in the same order),
followed by $w'_j$, followed by all other women in arbitrary order; we define the preference list of $m'_j$ to consist of $w_j$, followed by all other women in arbitrary order.

Similarly to the proof of
Lemma \ref{findstabbl-in-findstab}, we have that $w$ is single in some marriage $\marriage$ between $W$ and $M$ that is stable \wrt\ $\prefs{W}$ and $\prefs{M}$
iff $w$ and $m'$ are married in some marriage (a corresponding ``supermarriage'' of $\marriage$) between $W\cup W'$ and $M\cup M'$ that is stable \wrt\ $\prefs{W\cup W'}(\prefs{W})$ and $\prefs{M\cup M'}(\prefs{M})$.
Additionally, by Theorem \ref{thm:rural-hospitals} (in conjunction with Theorem \ref{thm:gs62}), we have: $w$ is single in some marriage between $W$ and $M$ that is stable \wrt\ $\prefs{W}$ and $\prefs{M}$ $\Longleftrightarrow$
$w$ is single in every marriage between $W$ and $M$ that is stable \wrt\ $\prefs{W}$ and $\prefs{M}$
$\Longleftrightarrow$
$w$ and $m'$ are married in every marriage between $W\cup W'$ and $M\cup M'$ that is stable \wrt\ $\prefs{W\cup W'}(\prefs{W})$ and $\prefs{M\cup M'}(\prefs{M})$.
\end{proof}


\section{Verifying the Output of a Given Stable Marriage Mechanism}\label{mech-verification}

As noted in Section \ref{sec:summary}, while our lower bound for verifying the stability of a given marriage is tight, we do now know whether our lower bound for finding a stable marriage is tight as well. We note that we do not even know a tight lower bound for verifying whether a given marriage is the $M$-optimal stable marriage.

\begin{open}\label{m-optimal-complexity-tight}
What is the worst-case complexity of verifying whether a given marriage is the $M$-optimal stable marriage?
\end{open}

As in the case of Open Problem~\ref{findstab-comparisons}, we do not have any $o(n^2 \log n)$ algorithm for verification of the $M$-optimal stable marriage,
even randomized and even in the strong two-party communication model, 
nor do we have any $\omega(n^2)$ lower bound, even for deterministic algorithms and even in the simple comparison model.

In this section, we the derive a $\Omega(n^2)$ lower bound for verification of the $M$-optimal stable marriage. In fact, we show this lower bound not only for verifying the \mbox{$M$-optimal} stable marriage, but also for verifying the output of any other stable marriage mechanism.

\begin{Definition}[Stable Marriage Mechanism]
Let $n\in\mathbb{N}$, let $W$ and $M$ be disjoint sets \sut\ $|W|=|M|=n$.
A \dft{stable marriage mechanism} is a function $f$ from \mbox{$\fullprefs(W,M)\times\fullprefs(M,W)$} to the set of perfect marriages between $W$ and $M$, \sut\ for every $\prefs{W}\in\fullprefs(W,M)$ and $\prefs{M}\in\fullprefs(M,W)$, the marriage $f(\prefs{W},\prefs{M})$ is stable \wrt\ $\prefs{W}$ and~$\prefs{M}$.
\end{Definition}

\begin{Example}[$M$-Optimal Stable Marriage Mechanism]
The function $\fmopt$, defined 
such that $\fmopt(\prefs{W},\prefs{M})$ is the \mbox{$M$-optimal} stable marriage \wrt\ $\prefs{W}$ and $\prefs{M}$, is a well-defined stable marriage mechanism
by Theorem \ref{thm:gs62}.
\end{Example}

\begin{corollary}[Complexity of Computing the Output of a Given Stable Marriage Mechanism]
By Theorem~\ref{thm:asm}, we have that for every stable marriage mechanism $f$, the worst-case randomized query complexity (as defined in Corollary~\ref{cor:asm}) as well as the worst-case communication complexity of computing $f$ is $\Omega(n^2)$.
\end{corollary}

\begin{theorem}[Complexity of Verifying the Output of a Given Stable Marriage Mechanism]
\label{mech-verification-complexity}
Let $n\in\mathbb{N}$, let $W$ and $M$ be disjoint sets \sut\ $|W|=|M|=n$, fix a stable marriage mechanism $f$ and let $\prefs{W}\in\fullprefs(W,M)$ and $\prefs{M}\in\fullprefs(M,W)$.
Let $\marriageid$ be the perfect marriage in which $w_i$ is married to $m_i$ for every $i$.
The worst-case randomized query complexity, as well as the worst-case randomized communication complexity,
of determining whether $f(\prefs{W},\prefs{M})=\marriageid$ is $\Omega(n^2)$.
\end{theorem}

Theorem~\ref{mech-verification-complexity} may be proven either via a direct application of the machinery of Section~\ref{sec:general-proofs}, or using the machinery of Section \ref{sec:smlb}, with Lemma~\ref{findstabbl-in-findstab} replaced by the following lemma.

\begin{lemma}\label{uniquebl-in-unique}
Let $n \in \mathbb{N}$, and let $W=\{w_1,\ldots,w_n\}$, $M=\{m_1,\ldots,m_n\}$, $W'=\{w'_1,\ldots,w'_n\}$ and $M'=\{m'_1,\ldots,m'_n\}$
 be pairwise-disjoint sets, each of cardinality $n$.
Let $\marriageid$ be the perfect marriage between $W$ and $M$ in which $w_i$ is married to $m_i$ for every $i$, and let $\marriageid'$ be the perfect marriage
between $W\cup W'$ and $M\cup M'$ in which for every $i$, both $w_i$ is married to $m_i$ and $w'_i$ is married to $m'_i$.
There exist functions $\prefs{W\cup W'}:\allprefs(W,M)\rightarrow\mbox{$\fullprefs(W\cup W',M\cup M')$}$ and $\prefs{M\cup M'}:\allprefs(M,W)\rightarrow\mbox{$\fullprefs(M\cup M',W\cup W')$}$ \sut\ for every
$\prefs{W} \in \allprefs(W,M)$ and $\prefs{M} \in \allprefs(M,W)$, both of the following hold.
\begin{enumerate}[label={\alph*.}]
\item\label{uniquebl-in-unique-stable}
If $\marriageid$ is the unique stable marriage \wrt\ $\prefs{W}$ and $\prefs{M}$, then $\marriageid'$ is the unique stable marriage \wrt\ $\prefs{W\cup W'}(\prefs{W})$ and $\prefs{M\cup M'}(\prefs{M})$.
\item\label{uniquebl-in-unique-unstable}
If $\marriageid$ is unstable \wrt\ $\prefs{W}$ and $\prefs{M}$, then $\marriageid'$ is unstable \wrt\ $\prefs{W\cup W'}(\prefs{W})$ and $\prefs{M\cup M'}(\prefs{M})$.
\end{enumerate}
\end{lemma}

\begin{proof}
We define $\prefs{W\cup W'}(\prefs{W})$ and $\prefs{M\cup M'}(\prefs{M})$ as in Lemma~\ref{findstabbl-in-findstab}, only with $M'$ appearing sorted by~$j$ (as opposed to in arbitrary order) on the preference lists of $W'$, and with $W'$ appearing sorted by $i$ (as opposed to in arbitrary order) on the preference lists of~$M'$. By Lemma~\ref{findstabbl-in-findstab}, we have both that \ref{uniquebl-in-unique-unstable} holds, and that if $\marriageid$ is the unique stable marriage \wrt\ $\prefs{W}$ and $\prefs{M}$, then it is a submarriage of every marriage that is stable \wrt\ $\prefs{W\cup W'}(\prefs{W})$ and $\prefs{M\cup M'}(\prefs{M})$; it is straightforward to show that every ``supermarriage'' of $\marriageid$, apart from $\marriageid'$, is unstable, thus proving \ref{uniquebl-in-unique-stable} as well.
\end{proof}

\begin{open}
Is there a stable marriage mechanism whose worst-case output verification complexity is $\Theta(n^2)$? Which stable marriage mechanisms have the lowest asymptotic worst-case output verification complexity?
\end{open}


\section{Nondeterminism}\label{nondeterministic}

All the lower bounds in this paper are based upon reductions to the well-studied communication complexity of the disjointness function. Since the disjointness function also has $\Theta(n)$ \dft{nondeterministic} communication complexity \cite{KN97}, it follows that all our lower bounds
apply not only to randomized communication complexity, but also to nondeterministic communication complexity. 
For nondeterministic communication complexity, the $\Omega(n^2)$ lower bound for finding a stable marriage is in fact tight (and so still is
the $\Omega(n^2)$ bound for verification of stability).  

For the decision problem of verifying the stability of a given marriage, the \dft{co-}nondeterministic communication complexity 
may be easily seen to be $\Theta(\log n)$.
In contrast, we note that the proof of Theorem~\ref{thm:related}(\ref{thm:related-ms}) may be easily adapted to show a $\Omega(n^2)$ lower
bound also for the co-nondeterministic communication complexities
of determining the marital status of a given couple. 

\begin{theorem}[Nondeterministic Communication Complexity of Determining the Marital Status of a Given Couple]
\label{isaremarried-nd-cc}
In the notation of Theorem~\ref{thm:related}(\ref{thm:related-ms}), both the nondeterministic and co-nondeter\-ministic communication complexities of determining whether $w$ and $m$ are married in some/every stable marriage are $\Omega(n^2)$.
\end{theorem}

For completeness, we show this lower bound also for the nondeterministic and co-nondeter\-ministic communication complexities of the intermediate problem of determining whether a given participant is single, which we presented in Appendix~\ref{isaremarried}. (This proof also yields Theorem~\ref{isaremarried-nd-cc} using the tools of that appendix and of Section~\ref{sec:smlb}.) These lower bounds follow from the results of Appendix~\ref{isaremarried} in conjunction with the following lemma.

\begin{lemma}[Is Participant Single? $\hookrightarrow\lnot$ Is Participant Single?]
\label{issingle-in-nissingle}
Let $n \in \mathbb{N}$, let $W$ and~$M$ be sets \sut\ $|W|=|M|=n$, and let $w'$ and $m'$ \sut\ $W$, $M$, $\{w'\}$ and $\{m'\}$ are pairwise disjoint; let $w \in W$.
There exist functions $\prefs{W\cup \{w'\}}:\allprefs(W,M)\rightarrow\mbox{$\allprefs(W\cup\{w'\},M\cup\{m'\})$}$ and $\prefs{M\cup\{m'\}}:\allprefs(M,W)\rightarrow\mbox{$\allprefs(M\cup\{m'\},W\cup\{w'\})$}$ \sut\ for every
$\prefs{W} \in \allprefs(W,M)$ and $\prefs{M} \in \allprefs(M,W)$, the following are equivalent.
\begin{itemize}
\item
$w$ is single in some marriage between $W$ and $M$ that is stable \wrt\ $\prefs{W}$ and $\prefs{M}$.
\item
$m'$ is married in every marriage between $W\cup\{w'\}$ and $M\cup\{m'\}$ that is stable \wrt\ $\prefs{W\cup\{w'\}}(\prefs{W})$ and $\prefs{M\cup\{m'\}}(\prefs{M})$.
\end{itemize}
\end{lemma}

\begin{proof}
To define $\prefs{W\cup \{w'\}}(\prefs{W})$, we define the preference list of $w$ as her preference list in~$\prefs{W}$ (in the same order), followed by $m'$; we define the preference list of every other woman in $W$ as her preference list in $\prefs{W}$ (in the same order and with $m'$ absent), and define the preference list of $w'$ to be empty (once again, $w'$ can be ignored, and is defined purely for aesthetic reasons --- so that $W\cup\{w'\}$ and $M\cup\{m'\}$ be of equal cardinality).
To define $\prefs{M\cup \{m'\}}(\prefs{M})$, we define the preference list of every man in $M$ as his preference list in $\prefs{M}$ (in the same order and with $w'$ absent); we define the preference list of $m'$ to consist solely of $w$.

Directly from definition of $\prefs{M\cup M'}$ and $\prefs{W\cup W'}$, we have that a natural bijection $\marriage\mapsto\marriage'$ exists between stable marriages \wrt\ $\prefs{W}$ and $\prefs{M}$ and stable marriages \wrt\ $\prefs{W\cup\{w'\}}(\prefs{W})$ and $\prefs{M\cup\{m'\}}(\prefs{M})$; this bijection is given by:
\begin{itemize}
\item
If $w$ is married in $\marriage$, then $\marriage'=\marriage$ (with $m'$ and $w'$ single in $\marriage'$).
\item
If $w$ is single in $\marriage$, then $\marriage'$ is the marriage obtained from $\marriage$ by marrying $w$ to~$m'$ (with $w'$ once again single in $\marriage'$).
\end{itemize}
Once again by Theorem \ref{thm:rural-hospitals} (in conjunction with Theorem \ref{thm:gs62}), and by the existence of this bijection, we have: $w$ is single in some marriage between $W$ and $M$ that is stable \wrt\ $\prefs{W}$ and $\prefs{M}$
$\Longleftrightarrow$ $w$ is single in every marriage between $W$ and $M$ that is stable \wrt\ $\prefs{W}$ and $\prefs{M}$ $\Longleftrightarrow$
$m'$ is married in every marriage between $W\cup\{w'\}$ and $M\cup\{m'\}$ that is stable \wrt\ $\prefs{W\cup\{w'\}}(\prefs{W})$ and $\prefs{M\cup\{m'\}}(\prefs{M})$.
\end{proof}

We note that the nondeterministic lower bound of $\Omega(n^2)$ for determining whether a given couple is married in some stable marriage, as well as the co-nondeterministic lower bound of $\Omega(n^2)$ for determining whether a given couple is married in every stable marriage (and both the nondeterministic and co-nondeterministic lower bounds of $\Omega(n^2)$ for determining whether a given participant is single in some/every stable marriage), is in fact tight. (Recall that we do not know whether any of these problems can be deterministically or even probabilistically solved using $o(n^2\log n)$~communication.)
The questions of a tight co-nondeterministic lower bound for the former problem and a tight nondeterministic lower bound for the latter remain open in all query models. We note that the latter problem may be solved by checking whether the pair in question is married in both the $M$-optimal stable marriage and the $W$-optimal stable marriage; a $\Oh(n^2)$-Boolean-queries algorithm (even a nondeterministic one) for verification of the $M$-optimal stable marriage (see Open Problem~\ref{m-optimal-complexity-tight} in Appendix~\ref{mech-verification}) would therefore also settle the question of the nondeterministic communication complexity of this problem.


\section{Optimality of Deferred Acceptance \wrt\texorpdfstring{\\}{ }Queries onto Women}\label{gale-shapley-opt}

Gale and Shapley's (1962) proof of Theorem \ref{thm:gs62} is constructive, providing an efficient algorithm for finding the $M$-optimal stable marriage. In this algorithm, men are asked queries of the form ``which woman is next on the preference list of man $m$  after woman $w$?'' (or alternatively, ``which woman does man $m$ rank at place $k$?''), while women are asked queries of the form ``whom does woman $w$ prefer most out of the set of men $\tilde{M}$?''; all of these queries require an answer of length $\Oh(\log n)$ bits.

Dubins and Freedman \cite{DF81} presented a variant of Gale and Shapley's algorithm, which runs in the same worst-case time complexity, but performs a significantly more limited class of queries, namely only pairwise-comparison queries, onto women. In Open Problem~\ref{findstab-comparisons} in the Introduction, we raise the question of a tight lower bound for the complexity of finding a stable marriage using only such queries \emph{for both women and men}. In this section, we show that regardless of how complex the queries onto the men may be, no algorithm for finding any stable marriage (and even no algorithm for verifying the stability of a given marriage, when input a stable marriage) that performs only pairwise-comparison queries onto women, may perform any less such queries onto them than Dubins and Freedman's variant of Gale and Shapley's algorithm (given the same preference lists). For the duration of this section, let $n\in\mathbb{N}$, let $W$ and $M$ be disjoint sets \sut\ $|W|=|M|=n$.

\begin{Definition}[Pairwise-Comparison Query]
A \dft{pairwise-comparison query} onto $W$ is a query of whether $m \!\succ_w\! m'$ for some given $w \!\in\! W\!$ and $m,m' \!\in\! M$.
\end{Definition}

\begin{Definition}[Men-Proposing Deferred-Acceptance Algorithm~\cite{DF81}]
The following algorithm is henceforth referred to as the \dft{men-proposing deferred-acceptance algorithm}:
The algorithm is initialized with all women and all men being \dft{provisionally single}, and concludes when no man is provisionally single.
The algorithm is divided into steps, to which we refer as \dft{nights}. On each night, an arbitrary provisionally-single man $m$ is chosen,
and serenades under the window of the woman $w$ ranked highest on his preference list among those who have not (yet) rejected him.
If $w$ is provisionally single,
then $m$ and $w$ are \dft{provisionally married} to each other.
Otherwise, i.e.\ if $w$ is already provisionally married to some man~$m'$, then
if $m \succ_w m'$, then $w$ rejects $m'$, who becomes provisionally single, and $w$ and $m$ are provisionally married to each other; otherwise, $w$~rejects $m$, who remains provisionally single.
The algorithm stops when no provisionally-single men remain, and the couples married by the output marriage are exactly those that are provisionally married when the
algorithm stops.
\end{Definition}

\begin{theorem}[\cite{DF81}]\label{m-optimal-df}
Let $\prefs{W}\in\fullprefs(W,M)$ and $\prefs{M}\in\fullprefs(M,W)$ be profiles of full preference lists for $W$ over $M$ and for $M$ over $W$, respectively.
The men-proposing deferred-acceptance algorithm stops after $\Oh(n^2)$ nights, and yields the $M$-optimal stable marriage.
\end{theorem}

\begin{Remark}
Let $\prefs{W}\in\fullprefs(W,M)$ and let $\prefs{M}\in\fullprefs(M,W)$.
All runs of the men-proposing deferred-acceptance algorithm (given $\prefs{W}$ and $\prefs{M}$) perform the same number of pairwise-comparison queries onto $W$.
\end{Remark}

\begin{theorem}[Optimality of Men-Proposing Deferred-Acceptance Algorithm \wrt\ Pairwise-Comparison Queries onto $W$]
\label{da-opt}
For any profiles~$\prefs{W}\in\fullprefs(W,M)$ and $\prefs{M}\in\fullprefs(M,W)$ of full preference lists for~$W$ over $M$ and for $M$ over~$W$, respectively, every algorithm for finding or verifying a stable marriage (for the latter --- when input any marriage that is stable \wrt\ $\prefs{W}$ and $\prefs{M}$) that only performs pairwise-comparison queries onto $W$ (and arbitrary
queries onto $M$), performs no less queries onto $W$ than the men-proposing deferred-acceptance algorithm, when input $\prefs{W}$
and $\prefs{M}$.
\end{theorem}

\begin{Remark}
An analogous result may similarly be shown to hold \wrt\ profiles of arbitrary preference lists, and finding/verifying a possibly-imperfect stable marriage.
\end{Remark}

\begin{Definition}
Let $\marriage$ be a perfect marriage between $W$ and $M$. By slight abuse of notation, we denote the woman married to a man $m \in M$ in $\marriage$ by~$\marriage(m)$ instead of $\marriage^{-1}(m)$.
\end{Definition}

\begin{proof}[Proof of Theorem~\ref{da-opt}]
Let $A$ be a run of the men-proposing deferred-acceptance algorithm \wrt\ $\prefs{W}$ and $\prefs{M}$, and let $B$ be a given run of an algorithm for finding/verifying a stable marriage \wrt\ $\prefs{W}$ and $\prefs{M}$.
Let $Q \subseteq W\times M^2$ be the set of triples $(w,m,m')$ \sut\ either the query of whether $m \succ_w m'$ was performed onto $W$ during $B$ and answered positively, or the query of whether $m' \succ_w m$ was performed onto $W$ during $B$ and answered negatively. By definition, at least $|Q|$ queries onto $W$ are performed during~$B$.
Let $\marriage$ be the $M$-optimal stable marriage \wrt\ $\prefs{W}$ and~$\prefs{M}$, i.e.\ the marriage output by~$A$.
Let $R=\set{(w,m) \mid \mbox{$w$ rejects $m$ during $A$}} \subseteq W\times M$.
By definition, we note that the number of queries onto $W$ during $A$ equals the number of rejections performed during $A$, and
so, as no woman rejects the same man twice, equals $|R|$. It is therefore enough to show that $|R| \le |Q|$ in order to complete the proof.

Let $\marriage'$ be the output of $B$ if it is a run of an algorithm for finding a stable marriage, or the input to $B$ if it is a run of an algorithm for verifying stability; either way,
$\marriage'$ a stable marriage \wrt\ $\prefs{W}$ and $\prefs{M}$. We claim that $w \succ_m \marriage'(m)$ for every $(w,m) \in R$. Indeed, as $m$ serenades under women's windows during $A$ in descending order of preference, the fact that $w$ rejects $m$ during $A$ implies $w \succ_m \marriage(m)$. By Theorem~\ref{m-optimal-df}, we thus have $w \succ_m \marriage(m) \succeq_m \marriage'(m)$, as claimed. As $B$ guarantees the stability of $\marriage'$, it must therefore ascertain that $\marriage'(w) \succ_w m$ for every~$(w,m) \in R$; therefore, as only pairwise-comparison queries are performed onto $W$ during $B$, there exists $m' \in M$ \sut\ $(w,m',m) \in Q$. We have thus shown that $R$ is contained in the projection of $Q$ over its first and last coordinates, and therefore $|R| \le |Q|$, and the proof is complete.
\end{proof}

\end{document}